\documentclass{article}
\usepackage[utf8]{inputenc}
\usepackage[left=2.5cm,right=2.5cm,
    top=2cm,bottom=2cm,bindingoffset=0cm]{geometry}
\usepackage[utf8]{inputenc} 
\usepackage[T2A]{fontenc}
\usepackage[english]{babel}
\usepackage{amsmath,amsthm,amssymb}
\usepackage{proof}
\usepackage{multicol}
\usepackage{hyperref}
\usepackage{float}
\floatstyle{boxed} 
\restylefloat{figure}
\usepackage{indentfirst}
\usepackage{amsfonts} 
\usepackage{csquotes}
\usepackage{stmaryrd}
\usepackage{cases}
\usepackage{xcolor}
\usepackage{authblk}
\usepackage{apacite}

\usepackage{tikz}
\tikzstyle{worldm}=[circle,draw,minimum size=1cm,fill=HSEblue!1, font=\sffamily\footnotesize]
\usetikzlibrary{positioning,arrows,calc}
\tikzset{
modal/.style={>=stealth',shorten >=1pt,shorten <=1pt,auto,node distance=2cm,
semithick},
world/.style={circle,draw,minimum size=1cm},
point/.style={circle,draw,inner sep=0.5mm,fill=black},
reflexive above/.style={->,loop,looseness=7,in=120,out=60},
reflexive below/.style={->,loop,looseness=7,in=240,out=300},
reflexive left/.style={->,loop,looseness=7,in=150,out=210},
reflexive right/.style={->,loop,looseness=7,in=30,out=330}
}

\newtheorem{definition}{Definition}
\newtheorem{theorem}{Theorem}
\newtheorem{example}{Example}
\newtheorem{lemma}{Lemma}
\newtheorem{remark}{Remark}
\newenvironment{notation}{\smallskip\par\noindent {\bf Notation.}}{\smallskip\par\noindent}

\newcommand{\folp}{\mbox{\it FOLP}^\Box_0}

\newcommand{\folpcs}[1]{\mbox{\it FOLP}^\Box_{#1}}  
\newcommand{\folpcsv}[1]{ \scriptsize \mbox{\it FOLP}^\Box_{CS}{(#1)}}
\newcommand{\cM}{{\mathcal M}}
\newcommand{\cE}{{\mathcal E}}
\newcommand{\cI}{{\mathcal I}}
\newcommand{\cL}{{\mathcal L}}

\newcommand{\term}[2]{#1\! :_{#2}\!}

\renewcommand{\models}{\Vdash}
\title{On Semantics of First-Order Justification Logic with Binding Modalities}
\author[1]{Tatyana L. Yavorskaya}
\author[1]{Elena L. Popova}
\affil[1]{Steklov Mathematical Institute of Russian Academy of Sciences, Moscow, Russia}
\date{}

\begin{document}

\maketitle

\abstract{We introduce the first order logic of proofs $FOLP^\Box$ in the joint language combining justification terms and binding modalities. The main issue is Kripke--style semantics for $FOLP^\Box$. We describe models for $FOLP^\Box$ in terms of valuations of individual variables instead of introducing constants to the language. This approach requires a new format of the evidence function.  This allows us to assign semantic meaning to formulas that contain free variables.  The main results are soundness and completeness of $FOLP^\Box$ with respect to the described semantics.}

\section*{Introduction}

In this manuscript we describe and study the hybrid first order logic with justifications and modality and semantics for it.

Justification logics were introduced by S. Artemov in \cite{a1995} (see \cite{a2001} for more details). The original justification logic described in \cite{a1995}  is  the Propositional Logic of Proofs, $LP$.  It is formulated in an extension of the propositional language by proof terms and an operator of the type {\it ``term:formula''} which represents proof--theorem relation in formal systems like Peano Arithmetic. Proof terms (called justification terms in our context) are constructed from proof variables and constants (called justification variables and constants now) with the help of operations on proofs (called operations on justifications). The first semantics described for $LP$ is an arithmetical interpretation in which proof terms are interpreted as codes of derivations in Peano Arithmetic and formulas correspond to arithmetical sentences. S. Artemov proved  completeness of $LP$ with respect to arithmetical semantics.

Another important feature of $LP$ proved by S. Artemov is its connection to propositional modal logic. Namely, the so-called forgetful projection of $LP$ which replaces all occurrences of proof terms in $LP$-formulas by the modality $\Box$ is modal logic $S4$. In other words, for each theorem of $S4$ one can recover justification terms for all occurrences of $\Box$ such that the resulting formula is provable in $LP$. This result is called {\it Realization of $S4$ in $LP$}; it shows that $LP$ is an explicit counterpart of $S4$, that is, proof terms represent explicitly the information hidden under the existential quantifier in provability reading of the modality $\Box$. Together with the arithmetical completeness theorem for $LP$, this result yields exact provability semantics for $S4$ and, therefore, for the intuitionistic logic.

Further investigations lead to finding explicit counterparts of other modal logics, the resulting logics were called Justification Logics. 
Artemov's method for proving realization for $S4$ was generalized to some of its subsystems in \cite{b2000}.
Two variants of justification counterpart of $S5$ were presented in \cite{aksh1999} and \cite{r2006}.
The uniform constructive method for realizing all normal modal logics formed by axioms {\bf d}, {\bf t}, {\bf b}, {\bf 4} and {\bf 5} is described in \cite{bgk2010}. All the justification logics corresponding to different axiomatizations of such modal logics (there are 24 of them for only 15 modal logics) are studied in \cite{goe2012} and \cite{gk2012}, where the more general method for proving realization was developed. In \cite{g2011}, \cite{sh2016} and \cite{f2020b} different justification counterparts for provability logic $GL$ were found and realization for them  was proven. Nonconstructive semantical method for proving realization was developed for $S4$ and $LP$ and extended to other pairs of a modal logic and its justification counterpart in works of M. Fitting, 
for example \cite{f2016} gives a description of this method and its application to the broad range of modal logics.

Different kinds of semantics for justification logic which are not based on provability were discovered later. The first non--arithmetical semantics for $LP$ was introduced in  \cite{m1997}. M. Fitting in \cite{f2005} presented Kripke--style semantics for a wide range of justification logics and  proved completeness of these logics with respect to it, that allowed to prove realization results  semantically (see also \cite{f2016}).  There are also other kinds of semantics.
Let us also mention topological semantics for several operation--free 
 fragment of the hybrid logic $S4+LP$ from \cite{an2008}, 
game semantics for $LP$ from \cite{r2009}, modular models described and studied in \cite{a2012} and \cite{ks2012}, subset models from 
\cite{ls2019,ls2021} and neighborhood models from \cite{g2024}. In this paper we work with Fitting semantics.

Our focus point in this manuscript is first-order justification logic.
In the context of provability, the first-order logic of proofs $FOLP$ was studied by S. Artemov and T. Yavorskaya  in \cite{y1998}, \cite{ay2001}. 
They described several variants of the appropriate language for $FOLP$. The essential point here was to provide syntactical constructions able to capture the
difference between global and local parameters in proofs.  Namely, for a formula $\Phi(x)$ with a free parameter $x$ one should differ between two propositions 
\begin{enumerate}
    \item 
``$t$ is a proof of a formula $\Phi$ which contains $x$ free'' and  
\item 
``$t$ is a proof of a formula $\Phi$ for a given value of $x$'', proposition with the parameter $x$. 
\end{enumerate}
For this purpose the justification operator in $FOLP$ is indexed with finite sets of individual variables which remain free in the proposition about provability. For the example above, 
$\term{t}{\varnothing} \Phi(x)$ means that ``$t$ is a proof of the formula $\Phi(x)$, and $x$ is a free variable of $\Phi$'' (here $x$ is a local parameter, it is bound in $\term{t}{\varnothing}\Phi(x)$), and $\term{t}{x} \Phi(x)$ means ``for the particular value of a parameter $x$, $t$ is a proof of the formula $\Phi(x)$'' ($x$ is a global parameter, it is free
in $\term{t}{x} \Phi(x)$). 

The Realization Theorem for first-order modal logic $S4$ in $FOLP$ was proven in \cite{ay2011} using cut-free sequent calculus for first order $S4$. 
Kripke style semantics for first-order logic of proofs was presented and studied by M. Fitting in \cite{f2011}, \cite{f2014}. He introduced  models for $FOLP$ 
(called Fitting models) and proved soundness and completeness of $FOLP$  with respect to them.  Roughly speaking, Fitting models are Kripke models with growing domain for first order $S4$ supplied with the evidence function assigning sets of possible world to each closed formula and justification term. In \cite{fs2020} constant domain semantics is described and the corresponding first order justification logic is found.

Similarly with the language of first order justification logic, the first order modal language can be supplied with the syntactical tools for dealing with  global and local parameters. So-called binding modalities $\Box_X$ which bind all variables other than variables from $X$  were introduced in \cite{ay2016}. 
In this paper the extension of $S4^b$ with binding modalities is presented and Gentzen-style calculus admitting cut-elimination is described. 
Further study of binding modalities can be found in \cite{f2020a}.

In the current work we introduce a first order logic of proofs $FOLP^\Box$ in the joint language combining justification terms and binding modalities. 
Such hybrid logics in propositional language for the first time  appeared  in provability context, namely, they combine $Box$ for provability in formal arithmetic and justification terms for arithmetical derivations. For the survey of these  logics see \cite{ys1997}, 
\cite{an2004}, \cite{goris2007} and \cite{goris2009}. Epistemic logic with justification were introduced and studied in \cite{an2004}, \cite{an2005}, \cite{k2010}.
The main issue of this work is Kripke-style semantics for $FOLP^\Box$. We apply Fitting models to deal not only with justification terms but also with modalities. Semantics introduced in our manuscript combine  models for $S4^b$ from \cite{ay2016} and a version of Fitting models for $FOLP$ close to  described in \cite{f2011}. 
Namely, instead of introducing individual constants to the  language as it is done in \cite{f2014}
we supply models with valuations of individual variables as in \cite{f2011}. This allows to assign semantic meaning to formulas that contain free variables.  The main results are soundness and completeness of $FOLP^\Box$ with respect to the described semantics. Similarities and differences of our models and those from \cite{f2011} are discussed in details after definition of a model.  

The structure of the paper is the following. In \ref{s_lang} we describe the language  and axioms of first order justification logic with binding modalities and prove simple facts about it. In \ref{s_model} we describe models for our logic, discuss several examples and prove soundness. In \ref{s_completeness} the completeness is proven.

\section{Language and Axioms}
\label{s_lang}

We use the  alphabet consisting of the following symbols.  
\begin{itemize}  
\item
$Var=\{x_1,x_2,\ldots\}$ is a set of individual variables,  
\item
$JVar=\{p_1,p_2,\ldots\}$ is a set of justification variables, 
\item
$JConst=\{c_1,c_2,\ldots\}$ is a set of justification constants. 
\item
$Pred=\{P_1^{n_1},P_2^{n_2},\ldots\}$ is a set of predicate symbols. Here the upper index denotes arity (the number of arguments) of the symbols, we assume that there are infinitely many symbols of any arity.
\item 
{\Large $\cdot$} and $+$ are symbols for binary operations on justifications, $!$ and $gen_x$ for each $x\in Var$ are  symbols for unary operations.
\item
unary modality $\Box_X$ and justification operator $\term{t}{X}$ for every finite subset $X$ of $Var$.
\item 
parentheses, quantifiers $\forall$, $\exists$ and boolean connectives
$\neg,\ \land,\ \lor,\ \to$, $\leftrightarrow$. убрала выделение
\end{itemize}
\begin{definition}
\em
We define the language $\mathcal{L}$ of first-order justification logic with binding modalities as follows.
 {\it Justification terms} are constructed from justification variables and constants by means of operations on justifications:  
$$t::= p_i \mid c_i \mid (t \cdot t) \mid (t + t)  \mid  !t \mid gen_x(t)$$ 
where $p_i \in JVar$, $c_i \in  JConst$, $x\in Var$. We denote the set of justification terms by $JTerm$.

{\it Formulas and free and bound occurrences of individual variables }  are defined by induction in the standard way:
\begin{itemize}
\item
$Q^n_i(x_1, \dots, x_n)$ where	$Q^n_i$ is a predicate symbol of arity $n$ and $x_i$ are individual variables is a formula, all occurrences of variables are free; 
\item
if $\Phi_1$, $\Phi_2$ are formulas then 
$\neg \Phi_1$, $(\Phi_1 \wedge \Phi_2)$, $(\Phi_1 \lor \Phi_2)$, $(\Phi_1\to \Phi_2)$, $(\Phi_1 \leftrightarrow \Phi_2)$ are formulas, occurrences of variables in the compound formulas are free if they are free in their components and bound if they are bound in the components;
\item
if $\Phi$ is a formula then $\forall x \Phi$, $\exists x\Phi$
are formulas, all occurrences of variables other then $x$ remain free or bound as they are in $\Phi$, all occurrences of $x$ are bound;
\item
if $\Phi$ is a formula then
$\Box_X \Phi$ and $\term{t}{X} \Phi$ are formulas
where $t$ is a justification term and  $X$ is a finite set of variables.
All bound occurrences of variables in $\Phi$ remain bound in $\Box_X\Phi$ and 
$\term{t}{X}\Phi$.
A free occurrence of a variable $x$ in a formula $\Phi$ remains
free  in  formulas $\Box_X\Phi$ and $\term{t}{X}\Phi$ if $x\in X$, otherwise it becomes bound. All occurrences of variables in the lower index $X$ in formulas $\Box_X\Phi$ and  
$\term{t}{X}\Phi$ are free.
\end{itemize}
\end{definition}

Note that justification terms do not contain individual variables. In $gen_x(t)$ $x$ is not an occurrence of a variable but just a lower index.  

We denote the set of formulas by $Fm$. 
By $FV(\Phi)$ we denote the set of free variables of the formula $\Phi$. Then
$FV(\Box_X\Phi) = FV(\term{t}{X}\Phi) = X$.

\begin{remark}
\em
We may restrict the language and admit only $\neg$, $\forall$ and $\land$ as basic symbols, then $\lor$, $\to$, $\leftrightarrow$ and $\exists$ are used as the standard abbreviations.
\end{remark}

\begin{example} 
\em
The occurrence of $x$ in formula $\Box_{\varnothing} P(x)$ is bound; similarly, $x$ is bound in $\term{t}{\varnothing} P(x)$. While the
occurrences of $x$ are free in both $\Box_{\{ x\}}P(x)$ and  $\term{t}{\{x\}} P(x)$. 
\end{example}

Note that the standard first-order modality $\Box\Phi$ corresponds to $\Box_{FV(\Phi)}\Phi$, so we keep the notation $\Box\Phi$ for $\Box_{FV(\Phi)}\Phi$. 

We will use following abbreviations: $\Box \Phi$ and $\term{t}{}\Phi$ for $\Box_{FV(\Phi)}\Phi$ and $\term{t}{FV(\Phi)}\Phi$, respectively; $\term{t}{xy}\Phi$ and $\Box_{xy}\Phi$ for $\term{t}{\{x,y \}}\Phi$ and $\Box_{\{x,y \}} \Phi$, respectively. We also use vector notation $\overrightarrow{X}$ for $x_1, \dots, x_n$.

There are two types of substitutions in the language $\mathcal{L}$,  
substitution of justification terms for justification variables and replacement of free individual variables by other individual variables. We use the same notation for them, namely, by $\Phi [p/t]$ we denote the result of substituting justification term $t$ for all occurrences of justification variable $p$ everywhere in $\Phi$, and by  
$\Phi [x/y]$ the result of substituting variable $y$ for all free occurrences of $x$ everywhere in $\Phi$. For the latter, we assume as usually that 
$x$ does not occur in the scope of quantifiers on $y$ in $\Phi$.

\subsection{Logic $\folpcs{CS}$}

\begin{definition}
\em
\label{def_axioms}
$\folpcs{0}$ has the following  axiom schemata.

\medskip\noindent
\begin{tabular}{ll@{\qquad\quad}ll}
(A0) &  classical axioms of first-order logic & & \\[7pt]
(A1) &  $\term{t}{X \cup \{y\} } \Phi \to \term{t}{X} \Phi, \ y \not \in FV(\Phi)$ & 
(A1$'$) &  $\Box_{X \cup \{y\}} \Phi \to \Box_X \Phi, \ y \not \in FV(\Phi)$\\[7pt]
(A2) &  $\term{t}{X} \Phi \to \term{t}{X \cup \{y\} } \Phi$  & 
(A2$'$) &  $\Box_X \Phi \to \Box_{X\cup \{y\}} \Phi$ \\[7pt]
(A3) & $\term{t}{X} \Phi \to \Phi$ & 
 (A3$'$) & $\Box_X \Phi \to  \Phi $ \\[7pt]
 (A4) &  $\term{t}{X} (\Phi \to \Psi) \to (\term{s}{X} \Phi \to \term{[t \cdot s]}{X} \Psi)$ & 
(A4$'$) & $\Box_X (\Phi \to \Psi) \to (\Box_X \Phi \to \Box_X \Psi)$ \\[7pt]
(A5) & $\term{t}{X} \Phi \to \term{[t + s]}{X} \Phi$  & &  \\
& $\term{s}{X} \Phi \to \term{[t + s]}{X} \Phi$ & & \\[7pt]
(A6) &  $\term{t}{X} \Phi \to \term{!t}{X} \term{t}{X} \Phi$ & 
(A6$'$) & $\Box_X \Phi \to \Box_X \Box_X \Phi$ \\[7pt]
(A7) &  $\term{t}{X} \Phi \to \term{gen_x (t)}{X} \forall x \Phi, \ x \not \in X$
&          
(A7$'$) &  $\Box_X \Phi \to \Box_X \forall x \Phi, \ x \not \in X$ \\[7pt]
(A8) &  $\term{t}{X}\Phi \to \Box_X \Phi$ & & \\[7pt]

\end{tabular}

Rules of inference:
\begin{multicols}{2}
\begin{itemize} 
	\item[(R1)] $\vdash \Phi, \ \Phi \to \Psi$ $\Rightarrow$ $\vdash \Psi$
	\item[(R2)] $\vdash \Phi$ $\Rightarrow$ $\vdash \forall x \Phi$
	\item[(R3)] $\vdash \Phi$ $ \Rightarrow$ $ \vdash \Box_\varnothing \Phi$
	\item[] modus ponens
	\item[] generalization
	\item[] necessitation
\end{itemize}
\end{multicols}
\end{definition}

\begin{remark}
\em

    \begin{itemize}
		\item
		The standard derivations in first order logic show that the following two Bernays' rules are derivable in $\folpcs{}$: if $x\not\in FV(\Phi)$ then   
		$$
		\vdash \Phi\to\Psi\ \Rightarrow\ \vdash \Phi\to\forall x \Psi
		\ \mbox{ and }\ 
		\vdash \Psi\to\Phi  \Rightarrow  \vdash \exists x \Psi\to \Phi.
		$$
		In what follows we use these rules along with the Generalization Rule.
    \item 
    The following generalization of necessitation rule is derivable in $\folpcs{0}$ with the help of axiom (A2$'$)
    $\vdash \Phi$ $ \Rightarrow$ $ \vdash \Box_X\Phi$ for any finite set of individual variables $X$. We use this generalization when needed. 
        \item 
        (A3) is derivable from other axioms:
        $$
        \begin{array}{ll}
        \term{t}{X}\Phi\to\Box_X\Phi & \mbox{axiom (A8)} \\
        \Box_X\Phi\to\Phi & \mbox{axiom (A3$'$)}
        \\
        \term{t}{X}\Phi\to \Phi & \mbox{by syllogism}\\
        \end{array}
        $$
        \item 
        (A7$'$) is derivable from other axioms:
        $$
        \begin{array}{ll}
        \Box_X\Phi\to\Phi & \mbox{axiom (A3$'$)}
        \\
        \Box_X\Phi\to\forall x\Phi & \mbox{by Bernays rule}
        \\
        \Box_X(\Box_X\Phi\to\forall x\Phi) & \mbox{by necessitation rule (R3) }
        \\
        \Box_X\Box_X\Phi\to\Box_X\forall x\Phi & \mbox{by normality axiom (A4$'$)}
        \\
        \Box_X\Phi\to\Box_X\Box_X\Phi & \mbox{axiom (A6$'$)}
        \\
        \Box_X\Phi\to \Box_X\forall x\Phi
         & \mbox{by syllogism}\\
        \end{array}
        $$
                \item
        $\Box_X  \Box_{X \cup \{ y\} } \Phi \to \Box_X \Phi$
        is a theorem of $\folpcs{0}$:
$$
        \begin{array}{ll}
        \Box_{X\cup \{y\}}\Phi\to\Phi & \mbox{axiom (A3$'$)}
        \\
        \Box_X(\Box_{X\cup \{y\}}\Phi\to \Phi) & \mbox{by necessitation rule (R3)}
        \\
        \Box_X(\Box_{X\cup \{y\}}\Phi\to \Phi) \to (\Box_X\Box_{X\cup \{y\}}\Phi\to \Box_X \Phi) & \mbox{by axiom (A4$'$)}
        \\
        \Box_X\Box_{X\cup \{y\}}\Phi\to \Box_X \Phi & \mbox{by modus ponens (R1)}
        \end{array}
        $$ 
        \item 
        If $y \not \in X$, then 
        $\folpcs{0}\vdash\Box_X \Phi \leftrightarrow\Box_X \forall y \Phi$. Implication ``left-to-right'' can be derived using $(A7')$; implication ``right-to-left'' is derived below:
$$
      \begin{array}{ll}
        \Box_X \forall y \Phi \to\Phi  & \mbox{from (A3$'$) and (A0)  by syllogism }
        \\
        \Box_X(\Box_X \forall y \Phi \to\Phi ) & 
        \mbox{by necessitation rule (R3)}
        \\
        \Box_X\Box_X\forall y \Phi \to\Box_X\Phi & \mbox{by normality axiom (A4$'$)}
        \\
        \Box_X\forall y \Phi \to\Box_X\Box_X\forall y\Phi  & \mbox{axiom (A6$'$)}
        \\
        \Box_X\forall y \Phi \to \Box_X\Phi
         & \mbox{by syllogism}\\
        \end{array} 
        $$           
    \end{itemize} 
\end{remark}

\begin{definition}[Constant Specification]
\em
{\it Constant Specification} is any set of formulas of the form
$\term{c}{\varnothing} \Phi$
where $c \in JConst$, $\Phi_i$ is an $\folpcs{0}$-axiom.

By $\folpcs{CS}$ we denote logic obtained from $\folpcs{0}$ by adding formulas from $CS$ as new axioms.
\end{definition}

The definitions of derivation and derivation  from hypothesis  for $\folpcs{0}$ and $\folpcs{CS}$ are standard with standard restrictions. Namely, in order to have Deduction Theorem, for derivation from hypothesis we assume that generalization is not applied to variables free in hypothesis and that necessitation rule is applied only to axioms of $\folpcs{0}$. Since $\folpcs{0}$ contains transitivity axiom $\Box_X\Phi\to\Box_X\Box_X\Phi$ and derives 
$\term{t}{X}\Phi\to\Box_X\term{t}{X}\Phi$, this restriction does not change the set of derivable formulas.
We write
$\vdash\Phi$ 
if $\Phi$ is derivable in $\folpcs{0}$
and $CS\vdash\Phi$ or $\vdash_{CS}\Phi$ 
if $\Phi$ is derivable in $\folpcs{CS}$.  Here we are allowed to use all inference rules from Definition \ref{def_axioms}.
For derivability from hypotheses $\Gamma$ we use notation $\Gamma\vdash\Phi$ for $\folpcs{0}$  
and $\Gamma,CS\vdash\Phi$ or $\Gamma\vdash_{CS}\Phi$  for $\folpcs{CS}$, here we allow to apply generalization rule only to variables not free in $\Gamma$ and to apply necessitation rule only to axioms.


\begin{example} \em
One can derive the following formulas in $\folpcs{}$ 
with the appropriate constant specification.
\begin{itemize}
    \item $\term{t}{X} \Phi \leftrightarrow \Box_X \term{t}{X} \Phi$.
    \item $\term{t}{X}\Box_X \Phi \to \term{[a \cdot t]}{X} \Phi$ with $CS = \{ \term{a}{\varnothing} (\Box_X \Phi \to \Phi) \}$.
    \item $\term{t}{X} \Box_X \Phi \to \Box_X \term{[a \cdot t]}{X} \Phi$ with $CS = \{ \term{a}{\varnothing} (\Box_X \Phi \to \Phi) \}$.
    \item $\Box_X \term{t}{X} \Phi \to \term{[a \cdot !t]}{X} \Box_X  \Phi$ with  $CS = \{ \term{a}{\varnothing} (\term{t}{X} \Phi \to \Box_X \Phi) \}$.
    \item $\term{t}{X} \Phi \rightarrow \term{[a\cdot !t]}{X} \Box_X \Phi$ with $ CS = \{ \term{a}{\varnothing} (\term{t}{X} \Phi \to \Box_X \Phi) \}$.
    \item $\term{t}{X} \Phi \rightarrow \term{[b \cdot  (a\cdot !t)]}{X} \Box_X \Box_X \Phi$ with $CS = \{ \term{a}{\varnothing} (\term{t}{X} \Phi \to \Box_X \Phi), \term{b}{\varnothing} (\Box_X \Phi \to \Box_X  \Box_X \Phi) \}$.
\end{itemize}
\end{example}

\subsection{Internalization and Substitution in $\folpcs{}$}

Similarly with first order logic of proofs $FOLP$, our logic enjoys internalization property.

\begin{lemma}[Internalization] \em
\label{lm_int}
1. Assume that $p_1,\dots, p_n \in JVar$, $X_1, \dots, X_n$ are finite sets of individual variables and 
    $X = \bigcup\limits_{i=1}^{n} X_i$. 
    If 
    $\term{p_1}{X_1}\Phi_1, \dots, \term{p_n}{X_n}\Phi_n \vdash_{CS} \Phi, $
    then there exist a justification term $t(p_1, \dots,p_n)$ and a constant specification $CS'\supseteq CS$ s.t.
$$\term{p_1}{X_1}\Phi_1, \dots, \term{p_n}{X_n}\Phi_n \vdash_{CS'} \term{t(p_1, \dots,p_n)}{X}\Phi.$$

2. If $\Phi_1, \dots, \Phi_n \vdash_{CS} \Phi$ then there is a justification term $t(p_1, \dots p_n)$ and a constant specification $CS'\supseteq CS$ s.t. $\term{p_1}{X_1}\Phi_1, \dots, \term{p_n}{X_n}\Phi_n \vdash_{CS'} \term{t(p_1, \dots,p_n)}{X}\Phi$ for $X = \bigcup\limits_{i=1}^{n} X_i$.    

3. If $\vdash_{CS} \Phi$, then $\vdash_{CS'} \term{t}{\varnothing} \Phi$ for some justification term $t$ and constant specification $CS'\supseteq CS$. 
\end{lemma}

\begin{proof}
Let us prove (1). 
Induction on derivation of $\Phi$ from $\term{p_1}{X_1}\Phi_1, \dots, \term{p_n}{X_n}\Phi_n$. Initially $CS'=CS$. 

If $\Phi$ is an axiom of $\folp$, 
add   $\term{c}{\varnothing}\Phi$ to $CS'$ for some $c\in JConst$. Hence, $\vdash_{CS'}\term{c}{\varnothing}\Phi$  and using (A2) 
				$\vdash_{CS'}\term{c}{X}\Psi$.
If $\Phi \in CS$,  then $\Phi$ has the form $\term{c}{\varnothing}\Psi$. Using axiom (A6), we derive $\term{!c}{\varnothing}\term{c}{\varnothing}\Psi$ and using (A2) 
				$\term{!c}{X}\term{c}{\varnothing}\Psi$. 
If $\Phi$ is one of the hypothesis $\term{p_i}{X_i}\Phi_i$, then $X_i\subseteq X$. Applying (A6) and then (A2) yields $\term{!p_i}{X}\term{p_i}{X_i}\Phi_i$.

If $\Phi$ is obtained by modus ponens (R1) then by the induction hypothesis, $\vdash_{CS'} \term{t}{X}(\Psi \to \Phi)$ and $\vdash_{CS'} \term{s}{X} \Psi$.
Applying axiom (A4) yields $\vdash_{CS'} \term{[t\cdot s]}{X} \Phi$. 

If $\Phi$ is obtained by the generalization rule (R2), then it is applied to a variable which is not in $X$. 
Thus $\Phi = \forall x \Psi$. By the induction hypothesis, $\vdash_{CS'} \term{t}{X}\Psi$. Using axiom (A7) we derive $\term{gen_x(t)}{X}\forall x \Psi$. 
				
If $\Phi$ is obtained by necessitation rule, applied to an axiom, then $\Phi$ is of the form $\Box_{\varnothing} \Phi$ where $\Phi$ is an axiom.  
Extend $CS'$ by $\term{a}{\varnothing}\Phi$ and  $(\term{b}{\varnothing}(\term{a}{\varnothing} \Phi \to \Box_\varnothing \Phi) $ for some $a,b\in JConst$. Then by (A6) and (A4) we obtain
$\term{[b\cdot !a]}{\varnothing} \Box_\varnothing \Phi$.

To prove (2), assume that $\Phi_1, \dots, \Phi_n \vdash_{CS} \Phi$. By axiom (A3) we obtain $\term{p_1}{X_1}\Phi_1, \dots, \term{p_n}{X_n}\Phi_n \vdash_{CS} \Phi$. Therefore by (1) of the current  Lemma 
$\term{p_1}{X_1}\Phi_1, \dots, \term{p_n}{X_n}\Phi_n \vdash_{CS'} \term{t(p_1, \dots,p_n)}{X}\Phi$ for $X = \bigcup\limits_{i=1}^{n} X_i$. 

Note that (3) follows from (1) or (2) if we take the set of hypothesis empty.
\end{proof}

\begin{definition}
\em
A constant specification $CS$ is called {\it axiomatically appropriate}, if  
for each formula $\Phi$ if it is an axiom of $\folp$  then there is a justification constant $c$ s.t. $\term{c}{\varnothing}\Phi$ belongs to $CS$.
\end{definition}

\begin{remark}
\em
If $\Phi$ is axiomatically appropriate then in Lemma \ref{lm_int} one can take 
$CS'=CS$.
\end{remark}

\begin{definition}
\em
A constant specification $CS$ is called {\it  variant closed}, if for every substitution $\sigma$ with $Dom(\sigma)= FV(\Phi)$ one has $\term{c}{\varnothing}\Phi\in CS\ \Leftrightarrow\ \term{c}{\varnothing}\Phi\sigma\in CS$.
\end{definition}

The following Lemma is standard for first order justification logic.

\begin{lemma}[Substitution] \em
\label{lm_subst}
Assume that $\Phi$ is a formula, $\Gamma$ is a set of formulas. 
Let $\sigma$ be a  substitution of variables from $FV(\Gamma,\Phi)$ such that no collision of variables occurs in $\Psi\sigma$ for $\Psi\in\Gamma$ or in $\Phi\sigma$.
If  $CS$ is a variant closed constant specification, 
    $\Phi_1, \dots, \Phi_n \vdash_{CS} \Phi$, 
    then $\Phi_1\sigma, \dots, \Phi_n\sigma \vdash_{CS} \Phi\sigma$.
In particular, if  $\Phi_1, \dots, \Phi_n \vdash \Phi$, 
    then $\Phi_1\sigma, \dots, \Phi_n\sigma \vdash \Phi\sigma$.
\end{lemma}

\begin{remark}
\em
Internalization Lemma for a justification logic means that its own derivations can be represented by justification terms. This Lemma plays an important role in the proofs of Realization theorem, which says that for a modal logic $L$ and a justification logic $JL$ a formula $\Phi$ is a theorem of $L$ if and only if its realization $\Phi^r$ obtained by replacing each occurrence of $\Box$ by a justification term  is a theorem of $JL$. The brief list of the known results on realization is given in introduction of this paper. 

Realization for hybrid logics is a bit more delicate. There are two reasonable questions to ask about $\folp{}$. 
The first one is whether $\folp{}$ can be realized in $FOLP$, that is, if for a  theorem $\folpcs{CS}$ there exists a replacing of all occurrences of $\Box$ by justification terms, which transforms it into a theorem of 
$FOLP$. The second question is about connections of  $FOS4$ and $\folp$. If a formula is a theorem of $FOS4$ which occurrences of $\Box$ in it can be replaced by justification terms in such a way that the result is a theorem of $\folpcs{CS}$ with some $CS$.

For some propositional hybrid logics these questions were addressed, for example, 
in \cite{k2010},  \cite{goris2007},  \cite{goris2009} and \cite{g2012}. 
For logic $\folp{}$ realization is out of the scope of this paper and  requires further investigation.
\end{remark}

\section{Semantics}
\label{s_model}

\subsection{Definition of Fitting Models}

Fitting models for $\folp$ are Kripke models for first--order $S4$ supplied with evidence function for relation between justification terms and formulas. Note that formula $\Box\forall x \Phi\to\forall x \Box_x\Phi$ is derivable in $\folpcs{0}$, its validity in Kriple models corresponds to growing domains, so our models are based on transitive reflexive frames with growing domains. 
We need some definitions and abbreviations concerning assignment of objects to variables. 

\begin{definition}
\em

Given a set $D\not=\varnothing$ and a finite or countable alphabet $X$, we call {\it a valuation of $X$ in $D$} any  function $f$ from  $X$ to $D$. As  usually, we denote $X$ by $Dom(f)$ and the $f(X)$ by $Im(f)$. A valuation is called {\it finite} if its domain is finite.
It is convenient to admit 
the empty set as the only possible valuation for the empty $X$.

For valuations $f$ and $g$ with $Dom(f)=Y$, $Dom(g)=X$ if $g\subseteq f$ then 
we say that $g$ is
a {\it restriction of $f$ to $X$} and $f$ is an {\it extension of $g$ to $Y$}. 
\end{definition}

\begin{notation}
\begin{itemize}
    \item
		Let $f$ be a valuation with $Dom(f)\subseteq X$ for some finite or countable alphabet $X$. 
For arbitrary  $Y\subseteq X$ by $f\upharpoonright Y$
we denote restriction of $f$ to $Y\cap Dom(f)$. 
\item
For a finite valuation $f$ of a subset of $X$ in $D$ by $ext(f,D)$ we denote the set of all finite extensions $g$ of $f$ such that 
$Im(g) \subseteq D$.
\item
For a valuation $f$ of $X$ in $D$, any different  variables $x_1,\ldots,x_n\in X$ and arbitrary $d_1,\ldots,d_n\in D$ by $f^{x_1, \dots, x_n}_{d_1,\ldots, d_n}$ we denote  the valuation with the domain $X$ defined as follows:
 	    $$
f^{x_1, \dots, x_n}_{d_1, \dots, d_n} (y)=\begin{cases}
			d_i, & \text{if $y = x_i$}\\
            f(y), & \text{otherwise.}
		 \end{cases}
$$
\item
For a substitution $\sigma$ of variables $Y$ for variables $X$ and a valuation $g$ of $Y$  by $g\circ \sigma$ we denote their composition, that is, a valuation $f$ of $X$  such that 
for each $x\in X$ it holds that $f(x)=g(\sigma(x))$.
\end{itemize}
\end{notation}

\begin{definition}[Fitting model] 
\em
\label{fitting_model}
A Fitting Model for $\folpcs{0}$ is a tuple 
 	$$\mathcal{M} = (W, R, (D_w)_{w \in W}, I, \mathcal{E}),$$
 	where 
       \begin{itemize}
 		\item $(W,R)$ is an $S4$-frame, that is, $W \not = \varnothing$ is a set of possible world and $R \subseteq W \times W$ is a reflexive and transitive accessibility relation on $W$;
 		\item $\{D_w \not = \varnothing\ |\ \ w\in W\}$ is a family of domain sets. Abbreviation $D$ is used for $\bigcup\limits_{w \in W}D_w$.  We consider models with monotonic domains, that is,  $wRu$ implies  $D_w \subseteq D_{u}$;
   
         \item  $I$ is an interpretation function, that is,  for each $n$-place predicate symbol $P$ and $w \in W$  we have $I(P, w) \subseteq (D_w)^n$;

     \item 
$\mathcal{E}$ is an evidence function. For any justification term $t$, formula $\Phi$ and finite valuation $f$ of  individual variables  in $D=\bigcup\limits_{w\in W}D_w$,   
     $\mathcal{E}(t, \Phi, f)$ is a subset of $W$.  
			
 	\end{itemize}
 \noindent 

We require that the evidence function $\mathcal{E}$ satisfies the following  conditions:

 \begin{itemize}
\item adequacy condition: 

$w\in\cE(t,\Phi,f)$ implies $Im (f)\subseteq D_w$;
 
 \item substitution condition: 
 
 assume that $x_1,\ldots,x_n$ are distinct variables from $FV(\Phi)$, $y_1,\ldots,y_n$ are variables and $\sigma$ is a substitution that replaces $x_i$ by $y_i$, that is 
 $
  \sigma=\left(\begin{array}{cccc}
  x_1& x_2 & \ldots & x_n\\
  y_1& y_2 & \ldots & y_n\\
  \end{array}\right)
  $
  it holds that 
  $ \mathcal{E}(t,\Phi\sigma,f\upharpoonright FV(\Phi\sigma)) = \mathcal{E}(t,\Phi,f\circ \sigma )$;

\item $R$ closure condition: 

if $wRu$ and $w\in\mathcal{E}(t,\Phi,f)$ then $u\in\mathcal{E}(t,\Phi,f)$;

\item 
extension condition: 

for $w\in W$ and $g\in ext(f,D_w)$ if
$w\in\mathcal{E}(t,\Phi,f)$ then $w\in\mathcal{E}(t,\Phi,g)$;
\item 
restriction condition: 

$\mathcal{E}(t,\Phi,f) \subseteq\mathcal{E}(t,\Phi,f\upharpoonright FV(\Phi))$;  

\item 
$\cdot$ condition:

$\mathcal{E}(t, \Phi \to \Psi, f) \cap \mathcal{E}(s,\Phi, f)  \subseteq \mathcal{E}(t \cdot s, \Psi,f)$; 

\item + condition: 

$\mathcal{E}(t,\Phi,f) \cup \mathcal{E}(s, \Phi, f) \subseteq \mathcal{E}(t + s,\Phi,f)$;

\item ! condition: 

if $w\in\mathcal{E}(t, \Phi, f)$, then  $w\in\mathcal{E}(!t, \term{t}{X}\Phi, g)$ for $g\in ext(f, D_w)$ and $X$ such that $Dom(f)\cap FV(\Phi)\subseteq X\subseteq Dom(g)$;

\item $gen_x$ condition: 

$\mathcal{E}(t,\Phi,f) \subseteq \mathcal{E}(gen_x(t),\forall x\Phi, f)$
for $x\not\in Dom(f)\cap FV(\Phi)$.
 \end{itemize}
 \end{definition}

\begin{remark}
\em
\label{rm_e_inf}
Informally speaking, we think of a valuation $f$ as replacing individual variables by objects from $D$. Applying  $f$  to a formula $\Phi$ means that each variable  $x\in FV(\Phi)\cap Dom(f)$ is assigned the value $f(x)\in D$. One can consider $\Phi$ under valuation $f$ as a substitutional instance of $\Phi$. Note that not necessarily all free variables of $\Phi$ are assigned values, some of them remain parameters of a formula. Now, by $\mathcal{E}(t, \Phi, f)$ we informally mean the collection of possible worlds,
in which $t$ witnesses the substitutional instances of $\Phi$ obtained by applying $f$ to $\Phi$. 

We assume that at a world $w$ a term cannot witness an assertion about objects that are not from $D_w$, this leads us to the adequacy condition on $\mathcal{E}$. However, a term can witness an assertion with a free variable (that is, with a parameter), this implies that the same term witnesses all substitutional instances of this assertion. It is reflected in extension condition.

Then, if a formula $\Psi$ is obtained from a formula $\Phi$ by renaming of free variables $\sigma$, then substitution $f$ applied to $\Psi$ and  $f\circ\sigma$ applied to $\Phi$ result in the same assertion, therefore any term should witness them simultaneously. This observation gives us substitution condition.
\end{remark}

\begin{definition}
\em 
    We say that a Fitting model  $\mathcal{M}$ and its  evidence function $\mathcal{E}$ {\it meet constant specification $CS$} if $\term{c}{\varnothing} \Phi \in CS$ implies
$ \mathcal{E}(c, \Phi,\varnothing)=W$.
\end{definition} 

Now we are going to  define the truth relation {\it ``a formula $\Phi$ is true at the world $w$ of the model $\mathcal{M}$ under the valuation  $\nu$''}, the comments on its specific details are given in the remark following the definition. 

 \begin{definition}
 \em
     {\it A valuation $\nu$ for a model $\mathcal{M}$} is a valuation of the set of variables $Var$ in  the domain $D$ of $\mathcal{M}$.
Given a model $\mathcal{M}$ and a valuation $\nu$ for $\mathcal{M}$, we define the truth relation {\it ``$\Phi$ is true at the world $w$ of the model $\mathcal{M}$ under the valuation  $\nu$''}, denoted by  
$(\mathcal{M}, \nu), w \Vdash \Phi$,  by induction on $\Phi$.

\begin{itemize}
	\item $(\mathcal{M}, \nu), w \Vdash P(x_1, \dots, x_n) \Leftrightarrow \langle \nu(x_1), \dots, \nu(x_n)\rangle \in I(P, w)$  
	\item $(\mathcal{M}, \nu), w \Vdash \neg \Psi \Leftrightarrow \nu(FV(\Psi))\subseteq D_w \mbox{ and }(\mathcal{M}, \nu), w \not \Vdash \Psi$
	\item $(\mathcal{M}, \nu), w \Vdash \Psi \wedge \Theta \Leftrightarrow (\mathcal{M}, \nu), w \Vdash \Psi $ and $(\mathcal{M}, \nu), w \Vdash \Theta$
	\item $(\mathcal{M}, \nu), w \Vdash \forall x \Psi  \Leftrightarrow \forall a \in D_w((\mathcal{M}, \nu^x_a), w \Vdash \Psi)$ 
	\item 
	By $R(w)$ we denote the set $\{u\in W\ |\ wRu\}$ of all possible worlds accessible from $w$. Then
	
	$(\mathcal{M}, \nu), w \Vdash \Box_X \Psi \Leftrightarrow$ 
 \begin{enumerate}
 \item 
 $\nu(X)\subseteq D_w$ and 
     \item 
 $\forall u \in R(w)\ \forall d_1, \dots, d_n  \in D_{u}\ ((\mathcal{M},  \nu^{y_1, \dots, y_n}_{d_1, \dots, d_n}), u \Vdash \Psi),$ where $\{y_1, \dots, y_n\} = FV(\Psi) \setminus X$
 \end{enumerate}
	\item $(\mathcal{M}, \nu), w \Vdash \term{t}{X} \Psi  \Leftrightarrow$  
	\begin{enumerate}
 \item 
 $\nu(X)\subseteq D_w$ and 
		\item 
  $\forall u \in R(w) \ \forall d_1, \dots, d_n \in D_{u} \ ((\mathcal{M}, \nu^{y_1, \ldots, y_n}_{d_1, \dots, d_n}),  u\Vdash \Psi),$
  where $\{y_1, \dots, y_n\} = FV(\Psi) \setminus X$  
  \item  
  $w\in
  \mathcal{E}(t, \Psi, \nu\upharpoonright (FV(\Psi) \cap X))$
\end{enumerate}
We will omit writing a pair $(\mathcal{M}, \nu)$ when it is clear from context. 
\end{itemize} 
\end{definition}

\begin{remark} 
\label{rm_neg_inf}
\em
\begin{itemize}
\item
The standard property of Kripke models for first order modal logic is that each formula $\Phi$, true in a given world $w$, lives in $w$, that is, $\nu(x)\in D_w$ for each $x\in FV(\Phi)$. This also holds for our definition (proposition 1 of Lemma \ref{lm_models}). 

Note that in the truth condition for atomic formulas $\nu(x_i)$ for $i=1,\ldots,n$ belong to the domain $D=\bigcup\limits_{w\in W}D_w$. However,
$(\mathcal{M}, \nu), w \Vdash P(x_1, \dots, x_n)$ is equivalent to  
$\langle \nu(x_1), \dots, \nu(x_n)\rangle \in I(P, w)$ and $I(P, w)\subseteq D_w^n$, whence  $(\mathcal{M}, \nu), w \Vdash P(x_1, \dots, x_n)$ implies $\nu(x_i)\in D_w$ for $i=1,\ldots,n$. 

In case of negation we guarantee this proposition by adding the requirement for
$(\mathcal{M}, \nu), w \Vdash \neg \Psi$ to 
$(\mathcal{M}, \nu), w \not \Vdash \Psi$ for
$(\mathcal{M}, \nu), w \Vdash \neg\Psi$.
Without this additional requirement we may get $(\mathcal{M}, \nu), w \Vdash \neg \Phi$ because of the fact that $\nu(x)\not\in D_w$ for some $x\in FV(\Phi)$, which can lead to contradiction. Namely, consider a model $cM$, a valuation $\nu$ and $w\in W$, such that $\nu(x)\not\in D_w$ and $\cI(P,w)=D_w$ for a unary predicate letter $P$. Then $(\mathcal{M}, \nu), w \Vdash \neg \Psi$ and  $(\mathcal{M}, \nu), w \Vdash \neg P(x)$, contradiction. 
\item
Note that  in any model $\mathcal{M}$ for the valuation $f=\varnothing$ we also define $\mathcal{E}(t,\Phi,\varnothing)\subseteq W$.
 If $w\not\in\mathcal{E}(t,\Phi,\varnothing)$, then $w\not\Vdash \term{t}{\varnothing}\Phi$ for any $\nu$. In case $w\in\mathcal{E}(t,\Phi,\varnothing)$ for any $\nu$  we have $(\mathcal{M}, \nu),w\Vdash \term{t}{\varnothing}\Phi$  if and only if $(\mathcal{M}, \nu), w\Vdash \Box_\varnothing \Phi$.
\end{itemize}
\end{remark}

\begin{remark}
\label{rm_dif_f2011}
\em
Our definition of a model is similar to the definition from \cite{f2011}. Let us describe the main differences. 

Firstly, specification of an evidence function is different. 
In \cite{f2011} it has four arguments, namely,  a justification term, a formula, an infinite (defined on the set $Var$ of all individual variables) valuation and a finite set of variables $X$. In in our models infinite valuations are replaced by finite and the set of variables $X$ is removed from the list of arguments. The reason for our choice is to have finite objects as the arguments of evidence function and reduce the number of its parameters.

Secondly, we impose substitution conditions on evidence function that is absent in \cite{f2011}. Its role  is explained in Remark \ref{rm_e_inf}. 
\end{remark}

\begin{definition}
\label{def_validity}
\em
For a given model $\mathcal{M}$,  a valuation $\nu$ and a formula $\Phi$ we say that {\it $\Phi$ is true in $(\mathcal{M},\nu)$} and denote this as $(\mathcal{M},\nu)\models\Phi$ if for all possible worlds $w\in W$ from
$\nu(FV(\Phi))\subseteq D_w$ it follows that $(\mathcal{M},\nu),w\models\Phi$. 
We say that {\it $\Phi$ is true in $\mathcal{M}$} and denote this as $\mathcal{M}\models\Phi$ if for all valuations $\nu$ one has $(\mathcal{M},\nu)\models\Phi$. 
For a set of formulas $\Gamma$ we define $(\mathcal{M},\nu)\models\Gamma$
and $\mathcal{M}\models\Gamma$ as $(\mathcal{M},\nu)\models\Phi$ for all $\Phi\in\Gamma$ and $\mathcal{M}\models\Phi$ for all $\Phi\in\Gamma$ respectively.
\end{definition}

\subsection{Simple Properties and Examples}

\begin{lemma}\label{lm_models}\em
For each model $\mathcal{M}$, valuation $\nu$, possible world $w\in W$ and formula $\Phi$
\begin{enumerate}
\item
$(\mathcal{M},  \nu), w \Vdash \Phi$ implies $\nu(FV(\Phi))\subseteq D_w$,
\item 
for every $a\in D_w$ and variables $x\in FV(\Phi)$ and $y$  if $\nu(y)=a$ then
$$
(\cM,\nu),w\Vdash \Phi [x/y] \Leftrightarrow
(\cM,\nu^x_a),w \Vdash \Phi.
$$
\item 
     for every $a \in D_w$ and $y \not \in FV(\Phi)$, 
    $$(\mathcal{M}, \nu), w \Vdash \Phi \Leftrightarrow  (\mathcal{M}, \nu^y_a), w \Vdash \Phi. $$
\end{enumerate}
\end{lemma}

\begin{proof}
By induction on $\Phi$.
\end{proof}

The following lemma shows that for a reasonable list of conditions  we are able to construct an evidence function satisfying them.
\begin{definition}
\em
Assume that $(W,R,(D_w)_{w\in W})$ is an $S4$--frame with monotonic domains. Let $\cE_0$ a finite mapping which assigns a set of possible worlds to some triples $(p,\Phi,f)$ where $p\in JVar$, $\Phi$ is a formula of $\folp$ and $f$ is a finite valuation of $Var$ in $D$. We use the following notation
$$
\begin{array}{l@{\mbox{ is }}l}
JVar(\cE_0) & \{p\ |\ (p,\Phi,f)\in Dom(\cE_0) \mbox{ for some }
 \Phi, f \}\\
Fm(\cE_0) & \{\Phi\ |\ (p,\Phi,f)\in Dom(\cE_0) \mbox{ for some }
p,\  f \}\\
Var(\cE_0) & \{x\in Var\ |\ x\in FV(\Phi) \mbox{ for some } \Phi\in Fm(\cE_0)\}\\
\end{array}
$$
Such mapping $\cE_0$ is a  {\it basic evidence function} if 
\begin{enumerate}
\item
$(p,\Phi,f)\in Dom(\cE_0)$ and $g\subseteq f$ imply
$(p,\Phi,g)\in Dom(\cE_0)$,
\item
$\cE_0$ satisfy adequacy, restriction, extension, $R$--closure and substitution  conditions for evidence function on its domain.
\end{enumerate}
\end{definition}
\begin{lemma}
\label{lm_ev}
\em
Assume that $(W,R,(D_w)_{w\in W})$ is an $S4$--frame with a monotonic family of domain sets and $\cE_0$ is a basic evidence function.  
Then there exists an evidence function $\cE$ which is an extension of $\cE_0$.
\end{lemma}
\begin{proof}
Define $\cE(t,\Phi,g)$ by induction on the term $t$. 
For a justification variable $p$ if $(p,\Phi,g)\in Dom(\cE_0)$put $\cE(p,\Phi,g)= \cE_0(p,\Phi,g)$ , otherwise
$$
\cE(p,\Phi,g)= \left(\bigcup\{\cE_0(p,\Phi,f)\ |\ f\subsetneq g\}\right)\cup
\left(\bigcup\{\cE_0(p,\Psi,g\circ \sigma) \ |\ \Phi \mbox{ coincides with } \Psi\sigma \mbox{ for } \Psi\in Fm(\cE_0)\}\right)
$$
In particular, 
$\cE(p,\Phi,g)= \varnothing \ \mbox{ if }\ p\not\in JVar(\cE_0)$.
If a justification term $t$ is not a justification variable then define by induction
\begin{enumerate}
\item[E1]
$\cE(t_1\cdot t_2,\Phi,g)=\bigcup\{\cE(t_1,\Phi\to\Psi,f)\cap\cE(t_2,\Phi,f) |\ \Psi\in Fm(\folp), \ g\subseteq f\}$
\item[E2]
$\cE(t_1 + t_2,\Phi,g)=\cE(t_1,\Phi,g)\cup\cE(t_2,\Phi,g)$
\item[E3]
$
\begin{array}[t]{ll}
\cE(!t,\Phi,g)=
\bigcup\{\cE(t,\Psi,h) |\  Dom (h)\cap FV(\Psi) \subseteq X\subseteq Dom(g)\ \mbox{ and } h\subseteq g\} & \mbox{ if } \Phi \mbox{ is } \term{t}{X}\Psi\\[3pt]
\cE(!t,\Phi,g)=\varnothing & \mbox{ otherwise }\\ 
\end{array}
$
\item[E4]
$
\begin{array}[t]{ll}
\cE(gen_x(t),\Phi,g)=
\bigcup\{\cE(t,\Psi,g) |\  x\not \in Dom(g)\cap FV(\Psi)\} & \mbox{ if } \Phi=\forall x\Psi\\[3pt]
\cE(gen_x(t),\Phi,g)=\varnothing & \mbox{ otherwise }\\ 
\end{array}
$
\end{enumerate}
One can show that $\cE$ satisfies all conditions on evidence function by induction on the term $t$.  
\end{proof}

In all the examples below we use proposition 3 of Lemma \ref{lm_models} that allows us to define precisely valuation $\nu$ only on free variables of formulas in which we are interested, and take arbitrary values of $\nu$ for all other variables.

\begin{example} \em 
 Consider the following model in which formulas
 $\Box_xP(x) \rightarrow \Box_\varnothing P(x)$ and $\Box_xP(x) \rightarrow \Box_x \forall x P(x)$ are false. 
 It consists of one reflexive world with two-element domain, that is, $W = \{ w\}$, $R = \{ (w, w)\}$, $D_w = \{ a, b\}$. We take $I(P, w) = \{ a\} $. The truth value of these formulas does not depend on the evidence function, so we may take $\cE(t,\Phi,f)=\varnothing$ for all $t$, $\Phi$ and $f$ .  We take  $\nu(x) = a$. 

Since $(\cM,\nu),w\Vdash P(x)$ and $w$ is the only world accessible from $w$, by definition we conclude $(\cM,\nu),w\Vdash\Box_xP(x)$. 
However, $(\cM,\nu),w\not\Vdash \Box_\varnothing P(x)$ since  $(\cM,\nu^x_b),w\not \Vdash P(x)$ and $wRw$. Then,
$(\cM,\nu),w\not \Vdash \forall x P(x)$ whence $(\cM,\nu),w\not \Vdash \Box_x \forall x P(x)$. 

\begin{figure}[h]
\caption{One-point frame with two element domain.}
\label{modal1}
\begin{center}
\begin{tikzpicture}[modal]
\node[world] (w) [label=below:$w$] {a,b};
\path[->] (w) edge[reflexive above] (w);
\end{tikzpicture}
\end{center}
\end{figure}
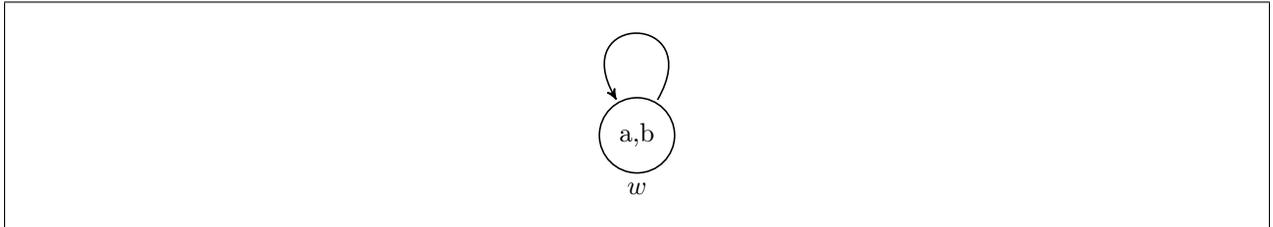
\end{example}

\begin{example} \em 
Consider the following formulas: 
\begin{eqnarray}
\label{form1} & \term{t}{\varnothing}P(x,y)\\
\label{form2} & \term{t}{x} P(x,y) \\ 
\label{form3} & \term{t}{y} P(x,y) \\ 
\label{form4} & \term{t}{xy} P(x,y)
\end{eqnarray}

In formula (\ref{form1}), all occurrences of individual variables are  bound. In (\ref{form2}) variable $x$ is free and $y$ is bound,  in (\ref{form3}) $x$ is bound and $y$ is free. In formula (\ref{form4}) all individual variables are free.

It is easy to check that for every model $\mathcal{M}$ formulas $\eqref{form1} \to \eqref{form2}\land \eqref{form3}$  and
$\eqref{form2}\lor \eqref{form3}\to \eqref{form4}$ are true in $\mathcal{M}$.

Indeed, for the first formula  by definition we should take a valuation $\nu$ and a possible world $w$ and show that if $\nu(x),\nu(y)\in D_w$ and  $(\mathcal{M},\nu),w\models \eqref{form1}$, then 
$(\mathcal{M},\nu),w\models \eqref{form2}\land \eqref{form3}$.
Assume that $\nu(x),\nu(y)\in D_w$ and $(\mathcal{M},\nu),w\models\eqref{form1}$, that is, 
$(\mathcal{M},\nu),w\models\term{t}{\varnothing}P(x,y)$. By definition of truth relation we get 
  $\forall u \in R(w) \ \forall d, e \in D_{u} \ (\mathcal{M}, \nu^{x, y}_{d,e}),  u\Vdash P(x,y)$ and  
  $w\in   \mathcal{E}(t, P(x,y), \varnothing)$. Since $\nu(x),\nu(y)\in D_w\subseteq D_u$ for all $u\in R(w)$, we have 
	$\forall u \in R(w) \ \forall d,e \in D_{u} \ (\mathcal{M}, \nu^{y}_{e}),  u\Vdash P(x,y)$ and $(\mathcal{M}, \nu^{x}_{d}),  u\Vdash P(x,y)$. Also by extension condition 
$w\in   \mathcal{E}(t, P(x,y), \nu\upharpoonright\{x\})$, 
$w\in   \mathcal{E}(t, P(x,y), \nu\upharpoonright\{y\})$, therefore 
$(\mathcal{M},\nu),w\models\term{t}{x}P(x,y)\land\term{t}{y}P(x,y)$. Validity of the second formula can be proven similarly.

Let us describe a model in which 
$\eqref{form4}\to\eqref{form2}\lor\eqref{form3}$ is false.
Consider the model $\cM_1$, based on the frame depicted in Figure \ref{modal1}, that is, $W = \{ w\}$, $R = \{ (w, w)\}$, $D_w = \{ a, b\}$. Let  $I(P, w) = D_w\times D_w$. Take 
$$
\begin{array}{c}
\mathcal{E}_0(t, P(x,y), \{\langle x,a\rangle,\langle y,b\rangle\} ) = \{ w\}\\
\mathcal{E}_0(t, P(x,y), \{\langle x,a\rangle\})=\mathcal{E}_0(t, P(x,y),\{\langle y,b\rangle\} ) = \varnothing\\
\end{array}
$$ 
and extend it to an evidence function by Lemma \ref{lm_ev}. 
Take valuation $\nu$ such that $\nu(x) = a$, $\nu(y) = b$.

For all $d,e \in D_w$ we have $(\cM_1, \nu^{xy}_{de}), w \Vdash P(x,y)$. Given that $w \in \cE(t, P(x,y), \nu\upharpoonright\{x,y\})$, formula \eqref{form4} is true in $(\cM_1, \nu)$. However,  $w \not \in \cE(t, P(x,y), \nu\upharpoonright\{x\})$. Thus, formula \eqref{form2} is false.  Similarly for  formula \eqref{form3}.

Now let us describe a model in which 
$\eqref{form2}\to\eqref{form3}$ is false. 
Consider the model $\cM_2$ identical to $\cM_1$ except for the evidence function, namely, here we take 
$$
\begin{array}{c}
\mathcal{E}_0(t, P(x,y), \{\langle x,a\rangle,\langle y,b\rangle\} ) = 
\mathcal{E}_0(t, P(x,y), \{\langle x,a\rangle\})=\{ w\}\\
\mathcal{E}_0(t, P(x,y),\{\langle y,b\rangle\} ) = \varnothing\\
\end{array}
$$ 
and extend $\mathcal{E}_0$ to the evidence function $\mathcal{E}$ by Lemma \ref{lm_ev}. Let valuation $\nu$ be such that $\nu(x) = a$, $\nu(y) = b$. 

For all $d,e \in D_w$ we have $(\cM_2, \nu^{xy}_{de}), w \Vdash P(x,y)$, therefore $(\cM_2, \nu^{x}_{d}), w \Vdash P(x,y)$ and 
$(\cM_2, \nu^{y}_{e}), w \Vdash P(x,y)$. Since $w \in \cE(t, P(x,y), \nu \upharpoonright \{ x\})$, formula (\ref{form2}) is true in $(\cM_2, \nu)$. However, $w \not \in \cE(t, P(x,y), \nu \upharpoonright \{ y\})$. Therefore, formula (\ref{form3}) is false.

Let us describe a model in which 
$\eqref{form2}\land\eqref{form3}\to\eqref{form1}$ is false. 
We take the model $\cM_3$ based on the same frame with interpretation $I(P, w) = \{  \langle a,a \rangle,$ $\langle a,b \rangle,$  $ \langle b,b \rangle \} $ and by Lemma \ref{lm_ev} extend $\mathcal{E}_0$ given by the equations
$$
\begin{array}{c}
\mathcal{E}_0(t, P(x,y), \{\langle y,b\rangle\} ) = 
\mathcal{E}_0(t, P(x,y), \{\langle x,a\rangle\})=\{ w\}\\
\cE(t, P(x,y), \varnothing)=\varnothing\\
\end{array}
$$ 
to the  evidence function. A valuation $\nu$ is such that $\nu(x) = a$, $\nu(y) = b$.
For all $d, e \in D_w$ we have $(\cM_3, \nu^y_d), w \Vdash P(x,y)$ and $(\cM_3, \nu^x_e), w \Vdash P(x,y)$. Given that $w \in \cE(t, P(x,y), \nu\upharpoonright\{x\}) \cap \cE(t, P(x,y), \nu\upharpoonright\{y\})$, we conclude that $(\cM_3, \nu)\models(\ref{form2})\land(\ref{form3})$.
However, $w \not \in \cE(t, P(x,y), \varnothing)$. Thus, $(\cM_3, \nu)\not\models(\ref{form1})$.  
\end{example}

\begin{example} \em
Let us construct a model in which formula $\forall x\ \term{t}{x} P(x)\to\Box\forall x P(x)$ is false. Note that the same reasonings applies to the corresponding modal formula $\forall x\ \Box_{x} P(x)\to\Box\forall x P(x)$ which is also false.

Note that this formula is true in any one-element model. Indeed, if $\forall x\ \term{t}{x} P(x)$ is true at the only world  $w$ of some model $\cM$, then by definition of the truth relation $P(x)$ is true at $w$ for each valuation of $x$ in the domain $D_w$. Since there are no possible worlds accessible from $w$ other than $w$ itself, we conclude that $\Box\forall x P(x)$
is true at $w$.

To falsify the given formula,
consider the model $\cM$ based on the two-element frame (Figure \ref{modal2}) with  $W = \{ w, u\}$,  $R = \{(w,w),(u,u), (w,u)\}$ and $D_w = \{a\}$, $D_u = \{a, b \}$. Take the interpretation $I(P, w) = I(P, u) = \{ a \}$ and the evidence function $\mathcal{E}$ extending $\mathcal{E}_0$ given  by the equations
$$
\begin{array}{c}
\mathcal{E}_0(t, P(x),  \varnothing) = \varnothing\\
\mathcal{E}_0(t, P(x),  \{\langle x,a \rangle\}) = \{w, u\}\\
\end{array}
$$
Choose  $\nu(x) = a$. 

For all $d \in D_w$ and for all $v \in R(w)$ we have $(\cM, \nu^x_d), v \Vdash P(x)$. Since  $w \in \cE(t, P(x), \nu \upharpoonright \{ x\})$, this gives  $(\cM, \nu), w \Vdash \forall x \term{t}{x} P(x)$. However, $(\cM, \nu), w \not \Vdash   \Box \forall x P(x)$   because $b \not \in I(P, u)$. 
Note that all formulas of the form $\term{s}{\varnothing} \forall x P(x)$ are false at $w$  for the same reason. 

\begin{figure}[h]
\caption{Two-point frame with increasing domains.}
\label{modal2}
\begin{center}
\begin{tikzpicture}[modal]
\node[world] (w) [label=below:$w$] {a};
\node[world] (v) [label=below:$u$,right=of w] {a,b};
\path[->] (w) edge (v);
\path[->] (w) edge[reflexive above] (w);
\path[->] (v) edge[reflexive above] (v);
\end{tikzpicture}
\end{center}
\end{figure}
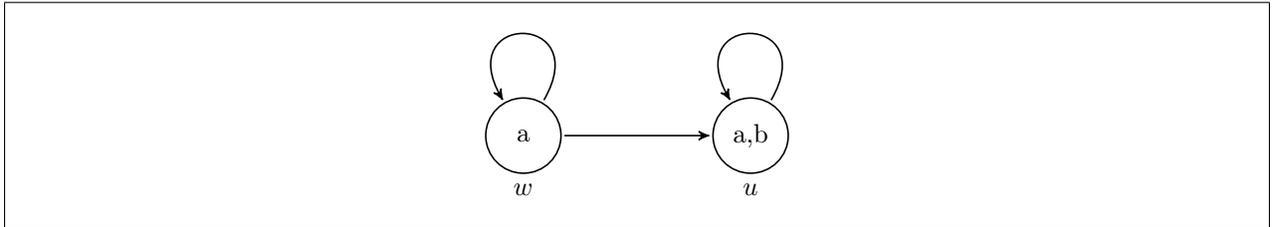
\end{example}

\begin{example} \em
Let us show how substitution condition for the evidence function works.  
For the first example consider  formula $\exists x (\term{t}{x} P(x))\to\exists y (\term{t}{y} P(y))$ which should be true in all models if our definition is relevant to the intuition of what principles are valid. 

For simplicity, consider one-element reflexive frame depicted by Figure \ref{modal1} and a model $\cM_1$ based on it with an arbitrary $I(P,w)\subseteq\{a,b\}$. By Lemma \ref{lm_ev} take evidence function $\cE$ extending mapping $\cE_0$ given  by the equations
$$
\begin{array}{c}
\cE_0(t, P(x),\{\langle x,a\rangle\})=\{w\},\ \ 
\cE_0(t, P(x),\varnothing)=\varnothing.
\end{array}
$$
Valuation $\nu$ is such that $\nu(x)=a$, $\nu(y)=b$. 
Formula $\exists x (\term{t}{x} P(x))$ is true at $w$ if  and only if $a\in I(P,w)$. Assume that this is the case, $a \in  I(P,w)$, $b \not\in  I(P,w)$. 
Note that formula $P(x)$ can be obtained by  applying the substitution $\sigma = \{ \langle y, x \rangle \}$ to $P(y)$, therefore $w\in \cE(t,P(y)\sigma,\nu\upharpoonright \{x\})$, thus by substitution condition on evidence function  $w \in \cE(t, P(y), \nu \circ \sigma)$. Given that $(\cM_1, \nu^y_a),w \Vdash P(y)$, we conclude $(\cM_1,\nu),w \Vdash \exists y \term{t}{y} P(y)$. 

For another example take a bit more nontrivial formula $\exists x\ \term{t}{x} P(x,x)\to \exists x\exists y\ \term{t}{xy} P(x,y)$ which also should be a valid principle if our semantics is relevant. 

As above, for simplicity we consider a model $\cM_2$ based on a one-element reflexive frame depicted by Figure \ref{modal1}. We  use Lemma \ref{lm_ev} and take any evidence function $\cE$ extending $\cE_0$ given by equations 
$$
\begin{array}{c}
\cE_0(t, P(x,x),\{\langle x,a\rangle\})=\{w\},\ \ 
\cE_0(t, P(x,x),\varnothing)=\varnothing\\
\end{array}
$$
Take  $I(P,w)$ to be an arbitrary subset of $\{a,b\}^2$  and valuation $\nu$ such that $\nu(x)=a$, $\nu(y)=b$.  

The formula $\exists x (\term{t}{x} P(x))$ is true at $w$ if  and only if $\langle a,a\rangle\in I(P,w)$. Assume that this is the case, $I(P,w)=\{\langle a,a\rangle\}$. 
Note that the formula $P(x,x)$ coincides with $P(x,y) \sigma$ for the substitution $\sigma = \{\langle y, x \rangle,\langle x,x\rangle\}$. By the choice of the evidence function and valuation, $w\in\cE(t,P(x,x),\nu\upharpoonright\{x\})$. Note that $\nu\circ\sigma= \{\langle x,a\rangle,\langle y,a \rangle\}$ and by substitution condition  on evidence function $w \in \cE(t,P(x,y), 
\nu\circ\sigma)$. Hence, $(\cM_2, \nu), w \Vdash \exists x \exists y\ \term{t}{xy} P(x,y)$.
\end{example}

\begin{example}
\em
In order to illustrate how !--condition works, let us check validity of a particular instance of the axiom (A6) 
\begin{equation}
\label{eq111}
\term{t}{xy}P(x,z) \to \term{!t}{xy}\term{t}{xy}P(x,z)
\end{equation}
in a model $\cM_1$, based on a three-element frame (Figure \ref{model3}) with  $W = \{ w, u, v\}$,  $R = \{(w,w),$ $(u,u),$ $(v,v),$ $(w,u),$ $(w,v)\}$ and $D_w = \{a\}$, $D_u = \{a, b \}$, $D_v = \{a, c \}$. Let $I(P, w)$, $I(P, u)$ and $I(P, v)$ be arbitrary subsets of $D_w^2$, $D_u^2$ and $D_v^2$. By Lemma \ref{lm_ev}, we take the evidence function $\mathcal{E}$, extending $\mathcal{E}_0$ given  by the following equations
$$
\begin{array}{c}
\mathcal{E}_0(t, P(x,z),  \varnothing) = \varnothing\\
\mathcal{E}_0(t, P(x,z),  \{\langle x,a \rangle\}) = \{w, u, v\}\\
\end{array}
$$
Choose  $\nu(x)=\nu(y) = a$. 
\begin{figure}[h]
\caption{Three-point frame with increasing domains.}
\begin{center}
\begin{tikzpicture}[modal]
\label{model3}
\node[world] (w) [label=below:$w$]{$a$};
\node[world] (u) [label=below:$u$,left=of w]{$a,b$};
\node[world] (v) [label=below:$v$,right=of w]{$a,c$};
\path[->] (w) edge (v);
\path[<-] (u) edge (w);
\path[->] (w) edge[reflexive above] (w);
\path[->] (u) edge[reflexive above] (u);
\path[->] (v) edge[reflexive above] (v);
\end{tikzpicture}
\end{center}
\end{figure}
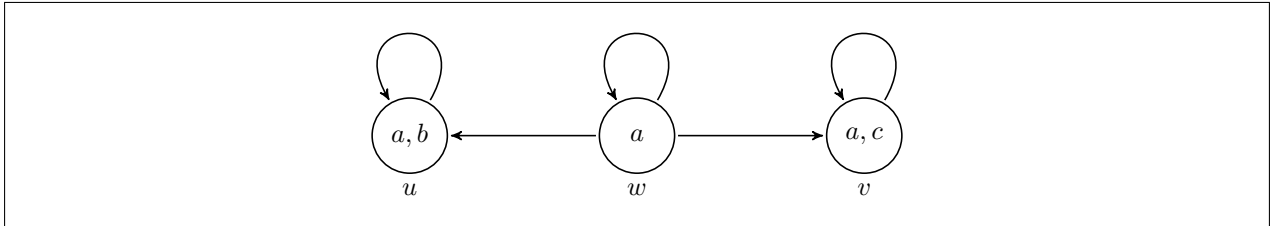

By the choice of evidence function, $w \in \cE(t, P(x,z), \nu\upharpoonright \{ x\})$. The formula $\term{t}{xy}P(x,z)$ is true in $w$ if and only if $(a,a) \in I(P, w)$, $\{ (a,a), (a,b)\} \subseteq I(P, u)$, $\{ (a,a), (a,c) \} \subseteq I(P, v)$. Assume that this is the case. 
By the ! condition on the evidence function, it holds that $w \in \cE(!t, \term{t}{xy}P(x,z),\nu\upharpoonright \{ x,y\} )$ since $\nu\upharpoonright \{ x,y\} \in ext(\{\langle x,a \rangle\}, D_w)$ and $\{ x\} \subseteq \{ x,y\}\subseteq Dom(\nu\upharpoonright \{ x,y\}) = \{x,y \}$.

It holds that $\nu(x),\nu(y) \in (D_u \cap D_v)$.
 For each $d \in D_u$ and each $e \in D_v$ we have  $(\cM_1, \nu^{z}_{d}),  u\Vdash P(x,z)$ and $(\cM_1, \nu^{z}_{e}),  v\Vdash P(x,z)$ respectively. Since $u, v \in \cE(t, P(x,z),  \{\langle x,a \rangle\})$, this gives $(\cM_1, \nu),  u\Vdash \term{t}{xy}P(x,z)$ and $(\cM_1, \nu),  v\Vdash \term{t}{xy}P(x,z)$.  
Therefore, we obtain $w \Vdash \term{!t}{xy}\term{t}{xy}P(x,z)$. Thus formula \eqref{eq111} is true at $w$. 

Now let us construct a model, in which the formula $\term{t}{x}P(x) \to \term{!t}{x}\term{t}{xy}P(x)$ is false. 

Consider a model $\cM_2$ based on the reflexive one-point frame depicted in Figure \ref{modal1}. Let $I(P,w) = D_w^2$. According to Lemma \ref{lm_ev}, we take the evidence function $\mathcal{E}$, extending $\mathcal{E}_0$ given  by the following equations
$$
\begin{array}{c}
\mathcal{E}_0(t, P(x),  \{\langle x,a \rangle\}) = \{w\}\\
\mathcal{E}_0(t, P(x),  \varnothing) = \varnothing\\
\end{array}
$$
We choose $\nu(x)=\nu(y)= a$.

By the construction of the model, we have $(\cM_2, \nu),w \Vdash P(x)$ and $w \in \cE(t, P(x), \nu\upharpoonright \{ x\})$. Hence, we obtain $(\cM_2, \nu),w \Vdash \term{t}{x}P(x)$ and   $(\cM_2, \nu),w \Vdash \term{t}{xy}P(x)$. 
Note that 
$\nu\upharpoonright(FV(\term{t}{xy}P(x)) \cap \{x\})=\{(x,a)\}$. 
So, for $(\cM_2, \nu),w \Vdash \term{!t}{x}\term{t}{xy}P(x)$ it remains to show, that $w\in\cE(t,\term{t}{xy}P(x),\{(x,a)\}$. But it is not necessarily the case. For example, for $\cE$ constructed in the proof of Lemma \ref{lm_ev} we have 
$$
\cE(!t,\term{t}{xy}P(x),\{(x,a)\})=
\bigcup\{\cE(t,P(x),h)\ |\  Dom (h)\cap FV(P(x)) \}\subseteq \{x,y\}\subseteq \{x\}=\varnothing.
$$ 
Thus, $(\cM_2, \nu),w \not \Vdash \term{!t}{x}\term{t}{xy}P(x)$ and the formula is false at $w$.  
\end{example}

\subsection{Soundness}

\begin{definition}
\em
    Let $\Gamma$ be a set of formulas, $\Phi$ be a formula. $\Phi$ is {\it a logical consequence of }$\Gamma$ (or $\Phi$ {\it logically follows from} $\Gamma$, notation $\Gamma\models\Phi$), if for every model $\cM=(W,R,(D_w)_{w\in W}, I, \cE)$, every valuation $\nu$ and possible world $w\in W$ from $(\cM,\nu),w\Vdash \Gamma$ and $\nu(FV(\Phi))\subseteq D_w$ it follows that $(\cM,\nu),w \Vdash \Phi$.
\end{definition}

\begin{theorem}[Soundness]
\em
	For each set of formulas $\Gamma$  and formula $\Phi$  
 $$\Gamma \vdash_{CS} \Phi \Rightarrow \Gamma \vDash_{CS} \Phi$$ 
\end{theorem}

\begin{proof}
 
Assume that $\mathcal{M}$ is a Fitting model meeting constant specification $CS$. 
We take $w\in W$. Let $\nu$ be any valuation such that $(\mathcal{M},\nu),w\Vdash \Gamma$ (whence $\nu(FV(\Gamma))\subseteq D_w$) and $\nu(FV(\Phi))\subseteq D_w$. Let us prove that $(\mathcal{M},\nu),w\Vdash \Phi$. 
Let $\Phi_1,\ldots,\Phi_n$ be a derivation of $\Phi$ from $\Gamma$.
By induction on $i$ we show that  
$(\mathcal{M},\nu'),w\Vdash \Phi_i$ for $i=1,\ldots,n$.
Since some of $\Phi_i$ may contain free variables that are not in $FV(\Gamma,\Phi)$ and $\nu$ does not necessarily return elements of $D_w$ for such variables, in order to make proof by induction on derivation possible we consider any valuation $\nu'$ which coincides with $\nu$ on variables for which $\nu(x)\in D_w$ and returns some $a\in D_w$ for other variables.  
 Since $\nu'$ coincide with $\nu$ on variables from $FV(\Gamma,\Phi)$, this proves the theorem. Without loss of generality we may assume that $\nu'$ coincides with $\nu$.

The case $\Phi_i\in\Gamma$ is trivial. Let us check that all axioms of $\folpcs{0}$ and formulas from $CS$ are true at all possible worlds.
Then in case $\Phi_i$ is an axiom of $\folpcs{0}$ or belongs to $CS$ or is obtained from an axiom by necessitation rule we have  $(\mathcal{M},\nu),w\Vdash \Phi_i$.

Axiom (A0). The case of propositional axioms is obvious.  Validity of axioms $\forall x \Phi(x)\to\Phi [x/y]$ and $\Phi [y/x]\to\exists x \Phi(x)$ is due to Lemma \ref{lm_models}(2). 


Axiom (A1). 
Suppose $y \not \in FV(\Phi)$. In this case 
$FV(\Phi) \setminus (X \cup \{ y\})= FV(\Phi) \setminus X$ and 
$FV(\Phi) \cap (X \cup \{ y\}) = FV(\Phi) \cap X$. Combining extension and restriction conditions we have $\mathcal{E}(t,\Phi,\nu \upharpoonright FV(\Phi) \cap (X\cup \{y\}) ) = \mathcal{E}(t,\Phi, \nu \upharpoonright FV(\Phi)\cap X ) $, therefore by definition of the truth relation  $(\mathcal{M}, \nu), w \Vdash \term{t}{X \cup \{y\} } \Phi$  and $(\mathcal{M}, \nu), w \Vdash \term{t}{X} \Phi$  are equivalent.


Axiom (A2). Assume that $\nu(X\cup\{y\})\subseteq D_w$ and $(\mathcal{M}, \nu), w \Vdash \term{t}{X} \Phi$, that is, 
 \begin{enumerate}
 \item $(\mathcal{M}, \nu^{y_1, \ldots, y_n}_{d_1, \dots, d_n}),  u\Vdash \Phi,$  
  for
$\{ y_1, \dots, y_n\} = FV(\Phi) \setminus X$, for all $u \in R(w)$ and  $  \{ d_1, \ldots, d_n\} \subseteq D_{u}$; 
    \item $w \in \mathcal{E}(t, \Phi, \nu \upharpoonright FV(\Phi) \cap X )$.
 \end{enumerate}
 There are two possible options  $y \in FV(\Phi) \setminus X$ or $y \not \in FV(\Phi) \setminus X$. In the latter case, adding a variable $y$ to the set $X$ does not change the set $\{ y_1, \dots, y_n\}$ and both conditions remain true for the formula $\term{t}{X \cup \{y\} } \Phi$, similarly to the proof for axiom (A1). In the first case without loss of generality we assume that $y$ coincides with $y_n$. Since 
 $\nu^{y_1, \ldots, y_{n-1}}_{d_1, \ldots, d_{n-1}}$ coincides with 
 $\nu^{y_1, \ldots, y_{n-1},y_n}_{d_1, \ldots, d_{n-1},\nu(y_n)}$, the first condition for the truth of 
 $\term{t}{X \cup \{y\} } \Phi$
 follows from the first condition for the truth of 
 $\term{t}{X} \Phi$. Then, by extension condition on evidence function, $w \in \mathcal{E}(t, \Phi, \nu \upharpoonright FV(\Phi) \cap (X \cup \{ y\} )      )$.  Therefore $(\mathcal{M}, \nu), w \Vdash \term{t}{X \cup \{ y\}} \Phi$.


Validity of (A3) is due to the reflexivity of the relation $R$. 
Among axioms that specify operations on justifications we consider (A4), (A6) and (A7) as the most nontrivial case and skip (A5) which can be treated similarly.

(Axiom A4). If $w\Vdash \term{t}{X} \Phi \to \Psi $ and $ w\Vdash \term{s}{X}\Phi$, then 
\begin{eqnarray}
\label{eq2}
    & w \in \cE(t, \Phi \to \Psi, \nu \upharpoonright FV(\Phi \to \Psi) \cap X),\ 
    w \in \cE(s, \Phi, \nu \upharpoonright FV(\Phi) \cap X), & \\[7pt] 
\label{eq3}    
& (\mathcal{M}, \nu^{y_1, \ldots, y_m}_{d_1, \dots, d_m}),  u\Vdash \Phi \mbox{ and }
    (\mathcal{M}, \nu^{y_1, \ldots, y_{m+k}}_{d_1, \dots, d_{m+k}}),  u\Vdash \Phi \to \Psi & \\
    & \mbox{ for } \{y_1, \ldots, y_m \} = FV(\Phi) \setminus X,\   
    \{y_1, \ldots, y_m,\ldots,y_{m+k} \} = FV(\Phi \to \Psi) \setminus X & \nonumber\\
    & \mbox{ for all } u \in R(w) 
\mbox{ and } d_1, \ldots, d_{m+k} \in D_u & \nonumber
    \end{eqnarray}

From \eqref{eq2} since $\nu \upharpoonright FV(\Phi) \cap X \subseteq \nu \upharpoonright FV(\Phi \rightarrow \Psi) \cap X$
by extension condition we have $w \in \mathcal{E}(s, \Phi, \nu \upharpoonright FV(\Phi \rightarrow \Psi) \cap X)$, whence $w \in \mathcal{E}(t \cdot s, \Psi, \nu \upharpoonright FV(\Phi \rightarrow \Psi) \cap X)$ due to $\cdot$ condition. Therefore we have by restriction condition
\begin{equation}
\label{eq3,5}
w \in \mathcal{E}(t \cdot s, \Psi, \nu \upharpoonright FV(\Psi) \cap X) 
\end{equation}
From \eqref{eq3} by Lemma \ref{lm_models}(3) one has $(\mathcal{M}, \nu^{y_1, \ldots, y_{m+k}}_{d_1, \dots, d_{m+k}}), w \Vdash \Phi$, therefore $(\mathcal{M}, \nu^{y_1, \ldots, y_{m+k}}_{d_1, \dots, d_{m+k}}), w \Vdash \Psi$. By Lemma \ref{lm_models}(3), 
\begin{equation}
\label{eq3,7}
(\mathcal{M}, \nu^{y_l, \ldots, y_{m+k}}_{d_l, \dots, d_{m+k}}), w \Vdash \Psi, \mbox{ where }  \{ y_l, \ldots, y_{m+k}\} = FV(\Psi) \setminus X.
\end{equation} 
Thus, from \eqref{eq3,5} and \eqref{eq3,7} we have  $(\mathcal{M}, \nu), w\Vdash  \term{t \cdot s}{X} \Psi$.


(Axiom A6). Suppose that $(\mathcal{M}, \nu), w \Vdash \term{t}{X} \Phi$, that is, 
\begin{eqnarray}
\label{eq06} & \nu(X)\subseteq D_w;\\
		\label{eq6} & (\mathcal{M}, \nu^{y_1, \ldots, y_n}_{d_1, \dots, d_n}),  u\Vdash \Phi  
  \mbox{ for }
\{ y_1, \dots, y_n\} = FV(\Phi) \setminus X, \mbox{ for all } u \in R(w) \mbox{ and }    \{ d_1, \ldots, d_n\} \subseteq D_{u};& \\ 
 \label{eq7}  & w \in \mathcal{E}(t, \Phi, \nu \upharpoonright (FV(\Phi) \cap X)).& 
\end{eqnarray}
In view of \eqref{eq06} in order to prove that $(\mathcal{M}, \nu), w \Vdash \term{!t}{X}\term{t}{X} \Phi$ we have to show that 
\begin{eqnarray}
\label{eq8} & (\mathcal{M}, \nu),  u\Vdash \term{t}{X}\Phi  \mbox{ for all } u \in R(w)& \\ 
 \label{eq9}  & w \in \mathcal{E}(!t, \term{t}{X}\Phi, \nu \upharpoonright X).& 
\end{eqnarray}
Take $f=\nu \upharpoonright (FV(\Phi) \cap X)$, $g= \nu \cap X$. 
From \eqref{eq7} and \eqref{eq06} we have 
$$
w \in \mathcal{E}(t, \Phi, f)\ \mbox{ and }\ 
FV(\Phi)\cap Dom(f) = FV(\Phi) \cap X \subseteq X= Dom(g),
$$
therefore by $!$ condition on evidence function  $w \in \mathcal{E}(!t, \term{t}{X}\Phi, g)$, that is, \eqref{eq9}. 
It remains to establish  \eqref{eq8}, that is, to show that $wRu$ implies $(\mathcal{M}, \nu),  u\Vdash \term{t}{X} \Phi$. Indeed, from \eqref{eq06} and monotonicity of domains we get $\nu(X)\subseteq D_u$.
We take arbitrary $v\in R(u)$ and $\{ d_1, \ldots, d_n\} \subseteq D_{v}$. By transitivity of $R$ we conclude $v\in R(w)$, therefore from \eqref{eq6}
 $(\mathcal{M}, \nu^{y_1, \ldots, y_n}_{d_1, \dots, d_n}),  v\Vdash \Phi$. Finally, by \eqref{eq7} and $R$--closure condition on evidence function, $u \in \mathcal{E}(t, \Phi, \nu \upharpoonright FV(\Phi) \cap X)$.


Axiom (A7). Suppose $(\mathcal{M}, \nu), w \Vdash \term{t}{X} \Phi$ and $x \not \in X$. Let $\{ y_1, \dots, y_n\} = FV(\Phi) \setminus X$. We have  
\begin{eqnarray}
\label{eq10}
&\forall u \in R(w) \forall d_1, \ldots, d_n \in D_{u}\ (\mathcal{M}, \nu^{y_1, \ldots, y_n}_{d_1, \dots, d_n}),  u\Vdash \Phi
& \\ 
\label{eq11}
&\mbox{ and } w \in \mathcal{E}(t, \Phi, \nu \upharpoonright (FV(\Phi) \cap X)).&
\end{eqnarray}
Since $x \not \in X$ we conclude $FV(\Phi) \cap X = FV(\forall x \Phi) \cap X$ and $x\not\in Dom(\nu \upharpoonright (FV(\Phi) \cap X))$. Thus from \eqref{eq11} by $gen_x$ condition on evidence function we have 
\begin{equation}
\label{eq20}
w \in \mathcal{E}(gen_x(t), \forall x \Phi, \nu \upharpoonright (FV(\forall x \Phi) \cap X)).
\end{equation}
It remains to show that 
\begin{equation}
\label{eq21}
\forall u \in R(w) \  \forall  d_1, \ldots, d_m\in D_u \ (\mathcal{M}, \nu^{z_1, \ldots, z_m}_{d_1, \dots, d_m}),  u\Vdash \forall x \Phi.
\end{equation}
where $\{z_1, \ldots, z_m\}=FV(\forall x\Phi)\setminus X$.
There are two possible cases. If $x \not \in FV(\Phi)$ then $m=n$ and  
$\{z_1, \ldots, z_m\}=\{y_1, \ldots, y_n\}$, use Lemma \ref{lm_models} (3) and \eqref{eq10}.  
If $x \in FV(\Phi)$ then $m=n-1$. Without loss of generality we assume that $x$ is $y_n$ and $\{z_1, \ldots, z_m\}=\{y_1, \ldots, y_{n-1}\}$. Then \eqref{eq21} follows from definition of truth for the  universal quantifier.  From (\ref{eq20}) and (\ref{eq21}) we have $(\mathcal{M}, \nu), w \Vdash \term{gen_x(t)}{X} \forall x \Phi$. 

Axiom  (A8). $\term{t}{X}\Phi \to \Box_X \Phi$ validity trivially follows from definition of truth relation. 

The proof of soundness of axioms (A1$'$)--(A7$'$) is a simplified version of the proof of soundness for (A1)--(A7). 

The induction step for Modus Ponens and generalization are standard. \textcolor{blue}{For Modus Ponens, assume that $\Psi$ is obtained from $\Phi$ and $\Phi\to\Psi$. Consider valuation $\nu'$, which coincides with $\nu$ on variables from $FV(\Psi)$ and returns some $a\in D_w$ for all the remaining variables. By the induction hypothesis, $(\mathcal{M}, \nu'), w \Vdash\Phi\to\Psi$ and $(\mathcal{M}, \nu'), w \Vdash\Phi$. Therefore $(\mathcal{M}, \nu'), w \Vdash\Psi$, whence $(\mathcal{M}, \nu), w \Vdash\Psi$.} For generalization,
if $\forall x \Phi$ is obtained from $\Phi$ \textcolor{blue}{and $\Gamma'\subseteq\Gamma$ is the set of hypotheses on which $\Phi$ depends. Then  $x$ does not occur free in  
$\Gamma'$. Therefore if $(\mathcal{M}, \nu), w \Vdash\Gamma$, then $(\mathcal{M}, \nu^x_d), w \Vdash\Gamma'$.} By the induction hypothesis $(\mathcal{M}, \nu^x_d), w \Vdash\Psi$
for any $d\in D_w$, therefore 
$(\mathcal{M}, \nu), w \Vdash\forall x\Psi$. 
\end{proof}

\section{Strong Completeness} 
\label{s_completeness}

\begin{theorem}[Strong Completeness]
\em
\label{st_compl}
 For each constant specification $CS$, set of closed formulas $\Gamma$  and a closed formula $\Phi$ 
    if $\Gamma \not\vdash_{CS} \Phi$ then there exists a  model $\cM$ meeting $CS$, a valuation $\nu$ and a possible world $w$ such that
    $(\cM, \nu), w\Vdash \Gamma$ but  $(\cM, \nu),w\not\Vdash \Phi$.
\end{theorem}

\begin{remark}
\em
In \cite{f2011} and \cite{f2014} a stronger completeness result  for $FOLP$ is proven.
Let us formulate it. We need the following definitions.

A model $\cM=(W,R,(D_w)_{w\in W},\cI,\cE)$ is called {\it fully explanatory}, if for every formula $\Phi$, every $w\in W$ and every valuation $\nu$ 
if $\nu(FV(\Phi))\subseteq D_w$ and $(\cM,\nu),w\models \Box_X\Phi$ 
for $X= FV(\Phi)$, then
$(\cM,\nu),w\models \term{t}{X}\Phi$ for some justification term $t$.

Two formulas are {\it variable variants} if they coincide up to the choice of free variables, that is, one formula is $\Phi(x_1/y_1,\ldots,x_n/y_n)$ and another is  $\Phi(x_1/z_1,\ldots,x_n/z_n)$ where $x_1,\ldots,x_n$ are all variables of $\Phi$. A constant specification  $CS$ is {\it variant closed}, if for variable variants of axioms $\Phi$ and $\Psi$ formulas $\term{c}{\varnothing}\Phi$ and $\term{c}{\varnothing}\Psi$ either are  both in $CS$ or both are not in $CS$.

It is proved in \cite{f2011} and \cite{f2014}, that if $CS$ is variant closed and axiomatically appropriate, then the canonical model for $FOLP_{CS}$ is fully explanatory, therefore, strong completeness with respect to fully explanatory Fitting models holds for $FOLP$. 
Similar result can be proven for $\folpcs{CS}$ with variable closed axiomatically appropriate $CS$. We do not give the proof here in order to keep the length of the paper reasonable. So, we do not to assume  
$CS$ to be variant closed or axiomatically appropriate in the completeness theorem. 
\end{remark}

We prove completeness of $\folpcs{CS}$ with respect to Fitting models via the canonical model construction. 

\subsection{Maximal $\exists$--Complete Sets}

We fix a countably infinite set of individual variables $Var^+$ such that $Var\cap Var^+ = \varnothing$. 
Let $V$ be a countably infinite subset of $Var^+$ with countably infinite complement to $Var^+$. 

We consider two types of extensions of the original language $\cL$ of $\folpcs{CS}$. 
The language $\mathcal{L}(V)$ is the extension of $\cL$ in which
variables from $V$ are allowed in formulas as additional individual variables. In particular, variables from $V$ are allowed as indexes in $gen_x$ and we may quantify on them. 

\begin{definition}
\em
    A formula $\Phi$ of the language $\cL(V)$ is {\it $V$--closed} if all occurrences of variables from $V$ in $\Phi$ are free, all occurrences of variables from $Var$ in $\Phi$ are bound and variables from $V$ are not allowed as indexes in $gen_x$.
\end{definition}

By $\mathcal{L}^h (V)$ we denote the set of $\cL(V)$--formulas in which
variables from $V$ are allowed in formulas as free variables only and cannot arise as indexes of $gen_x$.  
In fact formulas of $\mathcal{L}^h(V)$ are expressions of the form $\Phi\sigma$, where $\Phi$ is an $\cL$--formula and $\sigma$ is a finite substitution of variables from $Var$ by variables from $Var\cup V$. Note that $\mathcal{L}^h (V)$ contains all $V$-closed formulas. 

For each $\mathcal{L}^h (V)$--formula $\Phi$ by $FV(\Phi)$ denote the set of variables from $Var$ free in $\Phi$ and by $FV^+(\Phi)$
the set of variables from $Var^+$ free in $\Phi$. Notation $FV^h(\Phi)$ stand for
$FV(\Phi)\cup FV^+(\Phi)$.

We use abbreviations  $\folpcs{CS}(V)$ for the system in language $\cL^h(V)$ with $\cL^h(V)$--axioms similar with those of $\folpcs{CS}$ (see Definition \ref{def_axioms}). In other words, axioms of $\folpcs{CS}(V)$ have the form $A\sigma$ where $A$ is an axiom of $\folpcs{CS}$ and $\sigma$ is a finite substitution of  variables from $Var$ by variables from $V$. Remember that variables from $V$ are always free in formulas of  $\cL^h(V)$, thus, the generalization rule $R2$ cannot be applied to variables from $V$.

Let $S$ be a set of $V$--closed formulas.

\begin{definition} \em
 $S$ is {\it $\mathcal{L}(V)$--inconsistent using $CS$} if $S\vdash_{\folpcsv{V}}\bot$, 
 otherwise $S$ is {\it $\mathcal{L}(V)$--consistent using $CS$}.
\end{definition}

\begin{definition}\em
   $S$ is {\it $\mathcal{L}^h(V)$--maximal } if $\Phi \in S$ or $\neg\Phi \in S$ for each 
   $V$-closed formula $\Phi$.
\end{definition}

\begin{definition}\em
$S$ is  {\it $\exists$--complete w.r.t. $V$} if for each negated universal formula $\neg \forall x \Phi \in S$ there is some variable $v \in V$ (called witness) s.t. 
$ \neg \Phi[x / v] \in S$.   
\end{definition}

\begin{lemma}[Extension Lemma]
\label{lm_extension}
\em
Let $V_1\subseteq V_2$ be a subsets of $Var^+$, for which $V_2\setminus V_1$ is countable. Suppose that $S$ is a set of $V_1$--closed formulas that is 
$\mathcal{L}(V_1)$--consistent using $CS$. Then there exists a set  $S^+ \supseteq S$
of $V_2$-closed formulas such that $S^+$ is a $\mathcal{L}^h(V_2)$--maximal, $\mathcal{L}(V_2)$--consistent using $CS$ and $\exists$--complete w.r.t. $V_2$. 
\end{lemma}
\begin{proof}
The set of  $V_2$--closed formulas  is countable, 
    we numerate all $V_2$--closed formulas 
    as $\Phi_0, \Phi_1, \ldots$ and numerate elements of $V_2\setminus V_1$ as $v_1, v_2, \ldots$.  We define the sequence $S_0,S_1,S_2,\ldots$ of sets of $V_2$--closed formulas by recursion as follows:
\begin{enumerate}
\item
    $S_0 = S$. 
\item 
$S_{n+1}$ is 
\begin{itemize}
		\item
		$S_n \cup \{\neg \Phi_n\}$,  if $S_n \cup \{ \Phi_n \}$ 
   is   $\cL(V_2)$--inconsistent using  $CS$; 
	\item
     $S_n \cup \{ \Phi_n\}$, if $S_n \cup \{ \Phi_n \}$  
is $\cL(V_2)$--consistent using $CS$  and 
$\Phi_n$ is not of the form $\neg \forall x \Theta$; 
\item
$S_n \cup \{ \Phi_n, \neg \Theta[x/v_{i_n}] \}$,  if
$S_n \cup \{ \Phi_n \}$ is  $\cL(V_2)$--consistent using  $CS$  and 
$\Phi_n$  has the form  $\neg \forall x \Theta$  and 
             $v_{i_n}$  is the first variable in the list 
         $v_1,v_2,\ldots,$ not occurring in $S_n$ or $\Phi_n$.
				\end{itemize}
   \end{enumerate}
 Consider $S^+= \bigcup\limits_{n \in \mathbb{N}}S_n$. 
 
 It is easily seen that $S^+$ is $\mathcal{L}^h(V_2)$--maximal and $\exists$--complete w.r.t. $V_2$ by construction. Let us prove that $S_n$ is  $\mathcal{L}(V_2)$--consistent using $CS$ for each $n \in \mathbb{N}$. 
 Suppose that $S_n$ is $\mathcal{L}(V_2)$--consistent using $CS$ by induction hypothesis and $S_{n+1}$ is not. There are two possible cases. 
 
\begin{enumerate}
    \item $S_{n+1}$ is defined according to the first or the second rule. Then either $S_{n+1} = S_n \cup \{ \Phi_n\}$, where 
    $S_n \cup \{ \Phi_n \} \not \vdash_{\folpcsv{V_2}} \bot $, therefore $S_{n+1}$ is consistent by definition, or 
    $S_{n+1} = S_n \cup \{ \neg\Phi_n\}$, where 
    $S_n \cup \{ \Phi_n \} \vdash_{\folpcsv{V_2}} \bot $.
    In the later case we get $S_n \cup \{ \Phi_n\} \vdash_{\folpcsv{V_2}} \bot$ and $S_n \cup \{ \neg \Phi_n\} \vdash_{\folpcsv{V_2}} \bot$ (by the assumption that $S_{n+1}$ is inconsistent), whence $S_n$ is inconsistent, contradiction. 
    
    \item $S_{n+1} = S_n \cup \{ \Phi_n, \neg \Theta[x/v_{i_n}] \}$, where $\Phi_n$ is of the form $\neg \forall x \Theta$ and $S_n \cup \{ \Phi_n \} \not \vdash_{\folpcsv{V_2}} \bot$. If $S_{n+1}$ is inconsistent  using $CS$, then 
    $S_n\cup\{\neg \forall x \Theta, \neg \Theta[x/v_{i_n}]\} \vdash_{\folpcsv{V_2}} \bot$. Therefore, $S_n\cup\{\neg \forall x \Theta\} \vdash_{\folpcsv{V_2}} \Theta[x/v_{i_n}]$. Since $v_{i_n}$ has no occurrences in formulas of $S_n$ and $\Phi_n$ and $CS$ consists of closed formulas, we conclude $S_n\cup\{\neg \forall x \Theta\} \vdash_{\folpcsv{V_2}} \forall x \Theta$. 
Therefore $S_n \cup\{\neg\forall x \Theta\}\vdash_{\folpcsv{V_2}} \forall x \Theta\land \neg\forall x \Theta$, that is, $S_n \cup \{ \Phi_n \} \vdash_{\folpcsv{V_2}} \bot$, contradiction.
\end{enumerate}

Since  $S^+ = \bigcup\limits_{n \in \mathbb{N}}S_n$ and
$S_n$ is the increasing chain of $\mathcal{L}(V_2)$--consistent sets,  $S^+$ is $\mathcal{L}(V_2)$--consistent.
\end{proof}

The following lemma accumulates properties of maximal consistent sets of formulas and can be proven in the standard way.

\begin{lemma}
\label{lm_max_sets}
\em 
Let $S$ be a $\cL^h(V)$--maximal $\cL(V)$--consistent using $CS$ set of formulas. Then for all $V$--closed formulas $\Phi$ and $\Psi$  
\begin{enumerate}
     \item $\Phi \in S \Leftrightarrow S \vdash_{\folpcsv{V}} \Phi$,
     \item $\neg \Phi \in S \Leftrightarrow \Phi \not \in S$,
    \item $ \Phi \land \Psi \in S \Leftrightarrow \Phi \in S$ and $\Psi \in S$.
    \item
    If $S$ is also $\exists$--complete w.r.t. $V$, then $ \forall x \Phi \in S \Leftrightarrow \Phi [x/v] \in S$ for all $v \in V$.
\end{enumerate}
\end{lemma}

\subsection{Canonical Fitting Model}

\begin{notation}
For a set of $V$--closed formulas $S$, by $S^\#$ we denote the set of formulas 
$\forall y_1\ldots \forall y_n \Phi$ where $\Box_X \Phi \in S$ for some $X$ and $\{ y_1,\ldots ,y_n\} = FV(\Phi)$.
Note that $X\subseteq V$ since $S$ consists of $V$--closed formulas 
and 
$\forall y_1\ldots \forall y_n \Phi$ where $\{ y_1,\ldots ,y_n\} = FV(\Phi)$ is a $V$--closed formula.
\end{notation}

\begin{definition}
\em
\label{can_mod}
The canonical Fitting model 
    $\mathcal{M}^c = (W^c, R^c, (D_\Gamma^c)_{\Gamma \in W^c}, I^c, \mathcal{E}^c)$ is defined as follows.
    \begin{itemize}
        \item 
$W^c$ consists of all pairs  of the form 
$\Gamma = (S, V)$  where $V \subseteq  Var^+$ 
is countable with countable complement to $Var^+$ and $S$  is a set of $V$--closed formulas, that is $\mathcal{L}^h(V)$--maximal,  $\mathcal{L}(V)$--consistent using $CS$ and  $\exists$--complete w.r.t. $V$. 
        
We use notation $Form(\Gamma) = S$ and $Var(\Gamma) = V$ for each $\Gamma=(S,V) \in W^c$. 
        \item 
        For $\Gamma,\Delta\in W^c$ we define $\Gamma R^c \Delta \leftrightharpoons $
    \begin{enumerate}
    	\item $Var(\Gamma) \subseteq Var(\Delta) $ and
    	\item $Form(\Gamma)^\#\subseteq Form(\Delta)$.
    \end{enumerate}
    \item $D^c_\Gamma = Var(\Gamma)$; 
    \item $I^c$ is a canonical interpretation, that is, for each $n$--place predicate symbol $P$, possible world $\Gamma\in W^c$ and $v_1, \dots, v_n \in Var(\Gamma)$
    we define $\langle v_1, \dots, v_n \rangle \in I^c (P, \Gamma)$ if and only if $P (v_1, \dots, v_n) \in Form(\Gamma)$; 
    
    \item $\mathcal{E}^c$ is a canonical evidence function. 
		For a justification term $t$, an $\cL^h(Var^+)$--formula $\Phi$, a finite set $X\subseteq Var\cup Var^+$ and a finite valuation $f$  of $Var\cup Var^+$ 
		in $Var^+$ by $(\term{t}{X} \Phi)f$ we denote a $\cL^h(Var^+)$--formula obtained from $(\term{t}{X} \Phi)f$ by replacing all free occurrences of variables  $x\in Dom(f)$ by $f(x)$.  
We define
    $$
    \mathcal{E}^c(t, \Phi, f) =\{ \Gamma \in W^c \mid Im(f)\subseteq Var(\Gamma)
     \mbox{ and }\ 
          (\term{t}{X} \Phi)f\in Form(\Gamma)
     \mbox{ for }X = Dom(f)\cap FV^h(\Phi)
     \}.
    $$

    \end{itemize}
\end{definition}

\begin{lemma}
\em
\label{lm_canonical_model}
Canonical model $\cM^c$ is a Fitting model meeting $CS$ for the language $\cL^h(Var^+)$.
\end{lemma}

\begin{proof}
Let us check that $\cM^c$ satisfies all the conditions of Definition \ref{fitting_model}.
\begin{itemize}
\item  $R^c$ is reflexive. 

Suppose $\Box_X\Phi \in Form(\Gamma)$ for $\Gamma\in W^c$ (in particular, it means that $X\subseteq Var(\Gamma)\subseteq Var^+$). Since for $\{y_1, \ldots, y_n \} = FV(\Phi)$,
formulas $\Box_X \Phi \to \Box_X \forall y_1\ldots \forall y_n \Phi$ 
 and $\Box_X \forall y_1\ldots \forall y_n \Phi \to \forall y_1\ldots \forall y_n \Phi$ are derivable in   $\folpcs{0}(Var(\Gamma))$, we conclude that the $Var(\Gamma)$--closed formula
$\forall y_1\ldots \forall y_n \Phi$ is derivable from $Form(\Gamma)$, therefore by properties of maximal sets 
(see Lemma \ref{lm_max_sets})
  $\forall y_1\ldots \forall y_n \Phi \in Form(\Gamma)$. Thus  $\Gamma R^c \Gamma$ for all $\Gamma \in W^c$. 

\item $R^c$ is transitive. 

Assume that $\Gamma R^c \Delta$ and $\Delta R^c \Omega$. By definition of $R^c$, $Var(\Gamma) \subseteq Var(\Delta)$ and $Var(\Delta) \subseteq Var(\Omega)$, thus $Var(\Gamma) \subseteq Var(\Omega)$. Suppose that $\Box_X\Phi \in Form(\Gamma)$. Let us prove that $\forall y_1\ldots \forall y_n \Phi \in Form(\Omega)$ 
where $\{y_1,\ldots, y_n \} = FV(\Phi)$.
By Lemma \ref{lm_max_sets}, $\Box_X \Box_X \Phi \in Form(\Gamma)$ because $\Box_X \Box_X \Phi$ is derivable from $Form(\Gamma)$ using axiom (A6$'$). Hence, $\Box_X \Phi \in Form(\Delta)$ as $\Gamma R^c \Delta$. Then $\forall y_1\ldots \forall y_n \Phi \in Form(\Omega)$ since $\Delta R^c \Omega$. 
    
\item $\mathcal{M}^c$ has monotonic domains by definition of $R^c$. 

\item $\cE^c$ satisfies the adequacy condition. 

If $\Gamma\in\cE^c(t,\Phi,f)$, then $Im(f) \subseteq Var(\Gamma)$ by definition of $\cE^c$.

\item 
$\mathcal{E}^c$ satisfies the substitution condition. 

Suppose that $\Phi$ is a $\cL^h(V)$--formula, $\sigma$ is a substitution of its free variables (i.e., a mapping from  $FV^h(\Phi)$ to $Var$).  By definition of canonical evidence function,  
$\Gamma \in \mathcal{E}^c(t,\Phi\sigma,f\upharpoonright FV^h(\Phi\sigma))$  is equivalent to 
\begin{equation}
\label{eq300}
Im(f\upharpoonright FV^h(\Phi\sigma))\subseteq D^c_\Gamma  
\mbox{ and }
(\term{t}{X} \Phi \sigma)f \in Form(\Gamma) \mbox{ for } X= Dom(f)\cap FV^h(\Phi\sigma). 
\end{equation}
Since $FV^h(\Phi\sigma)= \sigma(FV^h(\Phi))$, 
for $Y = Dom(f \circ \sigma)$  we have 
$f(X)=(f\circ\sigma)(Y)$. Then $(\term{t}{Y} \Phi)(f \circ \sigma)$ coincides with $(\term{t\!}{X} \Phi\sigma)(f\upharpoonright FV(^h\Phi\sigma))$. Therefore \eqref{eq300} is equivalent to 
\begin{equation}
Im(f\circ\sigma )\subseteq D^c_\Gamma  
\mbox{ and }
(\term{t\!}{Y} \Phi)(f \circ \sigma) \in Form(\Gamma)
\end{equation}
that is, $\Gamma \in \mathcal{E}^c(t,\Phi,f\circ \sigma)$.

\item 
 $\mathcal{E}^c$ satisfies $R$ closure condition. 
 
 Suppose $\Gamma R^c \Delta$ and $ \Gamma \in\mathcal{E}^c(t,\Phi,f)$. 
From definition of $\mathcal{E}^c$ it follows that 
$Im(f)\subseteq Var(\Gamma)$ and
 $(\term{t \!}{X} \Phi)f \in Form(\Gamma)$ for $X= Dom(f)\cap FV^h(\Phi)$. Note that 
$(\term{t \!}{X} \Phi)f$ is a $Var(\Gamma)$--closed formula, therefore 
 $(\term{!t}{X} \term{t}{X} \Phi)f$ and 
$(\Box_{X} \term{t}{X} \Phi)f$ are $Var(\Gamma)$--closed too.
Then by Lemma \ref{lm_max_sets}, 
 $$
 \begin{array}{c}
 (\term{!t}{X} \term{t}{X} \Phi)f \in Form(\Gamma)   \text{ (use axiom A6), whence } \\
 (\Box_X \term{t}{X}\Phi)f \in Form(\Gamma) \text{ (use axiom A8).}\\
 \end{array}
 $$
By definition of $R^c$ we have $Im(f)\subseteq Var(\Delta)$ and from $FV((\term{t}{X}\Phi)f)=\varnothing$  we conclude that 
$(\term{t}{X} \Phi)f \in Form(\Delta)$. 
Thus, $\Delta \in \mathcal{E}^c(t,\Phi, f)$. 
 
\item
 $\mathcal{E}^c$ satisfies the extension condition. 
 
 Assume that $f$ and $g$ are  functions from finite subsets of $Var\cup Var^+$ to $Var^+$. Suppose that $\Gamma \in \mathcal{E}^c(t,\Phi,f)$ and $g\in ext(f, Var(\Gamma))$. 
Then for  $X=Dom(f)\cap FV^h(\Phi)$ we  have $(\term{t}{X} \Phi)f \in Form(\Gamma)$. Note that by definition of $\cL^h$--formulas $\Phi$ does not contain variables from $Var^+$ which do not belong to X. Since $g$ is an extension of $f$, we have $X\subseteq Dom (g)\cap FV^h(\Phi)$. Note that 
$$
\begin{array}{c}
Y=\{y_1,\ldots,y_k\} = (Dom (g)\cap FV^h(\Phi))\setminus X =(Dom (g)\cap FV(\Phi))\setminus X.
\end{array}
$$ 
Then $\Phi g$
coincides with $(\Phi f)(g\upharpoonright Y)$. 
Note that
 $\forall y_1\ldots \forall y_k ((\term{t}{X} \Phi)f\to (\term{t}{X\cup Y} \Phi)f)$ is a $Var(\Gamma)$--closed formula derivable in  $\folpcs{0}(Var(\Gamma))$, therefore $(\forall y_1\ldots \forall y_k \term{t}{X\cup Y} \Phi)f$ is derivable from $Form(\Gamma)$ in $\folpcs{0}(Var(\Gamma))$ whence 
  by Lemma \ref{lm_max_sets} $(\forall y_1\ldots \forall y_k \term{t}{X\cup Y} \Phi)f$ belongs to $Form(\Gamma)$. Taking into account that  $Im(g)\subseteq Var(\Gamma)$ we conclude that 
  $((\term{t}{X\cup Y} \Phi)f)(g \upharpoonright Y)$ belongs to $Form(\Gamma)$. 
Since $f$ and $g$ coincide on $X$, $((\term{t}{X\cup Y} \Phi)f)(g \upharpoonright Y)$ coincides with $(\term{t}{X\cup Y} \Phi)g$, therefore  
 $(\term{t}{X\cup Y} \Phi)g\in Form(\Gamma)$ whence $\Gamma \in \mathcal{E}^c(t,\Phi,g)$. 

\item 
$\mathcal{E}^c$ meets the restriction condition. 

 Suppose that $\Gamma \in \mathcal{E}^c(t,\Phi,f)$. Then $(\term{t}{X} \Phi)f\in Form(\Gamma)$ for $X= Dom(f)\cap FV^h(\Phi)$. 
 Since formulas  $(\term{t}{X} \Phi)f$ and $(\term{t}{X} \Phi)(f\!\upharpoonright\! FV^h(\Phi))$ coincide and $Dom(f\upharpoonright FV^h(\Phi))=
X$, we obtain $\Gamma \in \mathcal{E}^c(t,\Phi,f\upharpoonright FV^h(\Phi))$.

\item $\mathcal{E}^c$ meets $\cdot$ condition.

Suppose that  $\Gamma \in \cE^c(t, \Phi \to \Psi , f)$ and $\Gamma \in \mathcal{E}^c(s, \Phi, f)$. By definition of $\cE^c$, one has $Im(f) \subseteq Var(\Gamma)$ and 
\begin{eqnarray}
    (\term{t}{X} (\Phi \to \Psi))f \in Form(\Gamma) \mbox{ for } X = Dom(f) \cap FV^h(\Phi \to \Psi) \mbox{ and }\\
    (\term{s}{Y} \Phi)f \in Form(\Gamma) \mbox{ for } Y = Dom(f) \cap FV^h(\Phi).
\end{eqnarray} 
Since $Y\subseteq X$ and $Im(f) \subseteq Var(\Gamma)$, 
by Lemma \ref{lm_max_sets} and   axiom (A2) we have $(\term{s}{X} \Phi)f \in Form(\Gamma)$. Hence $(\term{t\cdot s}{X} \Psi)f \in Form(\Gamma)$ by Lemma \ref{lm_max_sets} and axiom (A4). 
Since $X\supseteq Dom(f)\cap FV^h(\Psi)$ by axiom (A1) we get $(\term{t\cdot s}{Z} \Psi)f \in Form(\Gamma)$ for $Z= FV^h(\Psi)$, therefore, $\Gamma \in \cE^c(t \cdot s, \Psi, f)$.

\item 
$\mathcal{E}^c$ meets $+$ condition. 

Suppose that  $\Gamma \in \mathcal{E}^c(t, \Phi , f)$, then $(\term{t}{X} \Phi)f \in Form(\Gamma)$ for $X = Dom(f)\cap FV^h(\Phi)$. By Lemma \ref{lm_max_sets} and axiom (A5),  $\{ (\term{t + s}{X} \Phi)f, (\term{s + t}{X} \Phi)f \} \subseteq Form(\Gamma)$, therefore, $\Gamma \in \mathcal{E}^c(t+s, \Phi , f)$ and $\Gamma \in \mathcal{E}^c(s+t, \Phi , f)$.

\item 
$\mathcal{E}^c$ meets $!$ condition. 

Assume that $\Gamma \in \mathcal{E}^c(t, \Phi, f)$, $g\in ext(f,D_\Gamma)$ and $FV^h(\Phi)\cap Dom(f)\subseteq X\subseteq Dom(g)$ for  $X\subseteq Var\cup Var^+$. 
By definition of $\cE^c$  we have $(\term{t}{X} \Phi)f \in Form(\Gamma)$ for $X = Dom(f)\cap FV^h(\Phi)$ and $Im(f)\subseteq Var(\Gamma)$.  Note that 
$(\term{t}{X} \Phi)f$ is $Var(\Gamma)$-closed, hence $(\term{!t}{X} \term{t}{X} \Phi)f$  is $Var(\Gamma)$-closed too.
Assume that $g\in ext(f, Var(\Gamma))$ and $FV(\Phi)\cap Dom(f)\subseteq Y\subseteq Dom(g)$, then  $(\term{t}{X} \Phi)g$ is $V$--closed. 
By Lemma \ref{lm_max_sets} using axioms (A6) and (A2) we consequently obtain  $(\term{t}{Y} \Phi)g \in Form(\Gamma)$ and
$(\term{!t}{Y} \term{t}{Y} \Phi)g \in Form(\Gamma)$. Since $Y$ coincide with $Dom(g)\cap FV^h(\term{t}{Y} \Phi)$, we conclude
$\Gamma\in\cE^c(!t,\term{t}{Y}\Phi,g)$. 

\item $\mathcal{E}^c$ meets $gen_x$ condition.   

Note that in the language $\cL^h(Var^+)$ variables from $Var^+$ are not allowed as indexes in $gen_x$, so $x\in Var$. Suppose that $\Gamma \in \mathcal{E}^c(t,\Phi,f)$.  
Take $X = Dom(f)\cap FV^h(\Phi)$, then $(\term{t}{X}\Phi)f \in Form(\Gamma)$ and $Im(f)\subseteq Var(\Gamma)$. 
We have to show that if $x\not\in X$ then $(\term{gen_x(t)}{Y}\forall x \Phi)f \in Form(\Gamma)$
for $Y=Dom(f)\cap FV^h(\forall x\Phi)=(Dom(f)\cap FV^h(\Phi))\setminus\{x\}=X$.
Since $x \not\in f(X)\subseteq Var^+$, by Lemma \ref{lm_max_sets} and axiom (A7) we obtain $\term{(gen_x(t)\!}{X \!} \forall x \Phi)f \in Form(\Gamma)$. 
Thus $ \Gamma \in \mathcal{E}^c(gen_x(t),\forall x\Phi, f)$. 

\item $\mathcal{E}^c$ meets $CS$ condition.   

Suppose $\term{c}{\varnothing}\Phi$ is $CS$. Since it is a closed formula (not containing any free variables either from $Var$ or from $Var^+$), it is $V$-closed for each $V$.
Since $\term{c}{\varnothing}\Phi$ is derivable in $\folpcs{CS}$, by Lemma \ref{lm_max_sets} we have $\term{c}{\varnothing} \Phi\in Form(\Gamma)$ for each $\Gamma\in W^c$, that is, for $f=\varnothing$ and
$X=Dom(f)\cap FV(\Phi)=\varnothing$ we have $Im(f)\subseteq Var(\Gamma)$ and $(\term{c}{X}\Phi)f\in Form(\Gamma)$, whence $\Gamma\in\cE^c(c,\Phi,\varnothing)$. This implies
$\cE^c(c,\Phi,\varnothing)= W^c$.
\end{itemize}
\end{proof}

\subsection{The Truth Lemma}

\begin{definition}
\em  A mapping $\nu$ from $Var\cup Var^+$ to $Var^+$ is called a {\it canonical valuation} if $\nu^c(x)=x$ for every $x \in Var^+$.
\end{definition}

\begin{lemma} 
\label{lm_truth}
\em
For the canonical model $\mathcal{M}^c$, any canonical valuation $\nu^c$, each $\Gamma \in W^c$ and each $Var(\Gamma)$--closed formula $\Phi$   
$$(\mathcal{M}^c,\nu^c), \Gamma \Vdash \Phi \Leftrightarrow \Phi \in Form(\Gamma).$$
\end{lemma}

\begin{remark}
\em
Since all formulas $\Phi$ considered in the Truth Lemma are $Var(\Gamma)$--closed and $Form(\Gamma)$ is $Var(\Gamma)$--maximal, we have $FV(\Phi)=\varnothing$, $\nu(FV^+(\Phi))\subseteq Var(\Gamma)$  and either 
$\Phi\in Form(\Gamma)$ or $\neg\Phi\in Form(\Gamma)$, therefore either
$(\mathcal{M}^c,\nu^c), \Gamma \Vdash \Phi$ or $(\mathcal{M}^c,\nu^c), \Gamma \Vdash \neg\Phi$. 
\end{remark}

\begin{proof} 
The proof is by induction on formula $\Phi$. \\  
\begin{itemize}
    \item $\Phi = P(x_1, \dots, x_n)$ where $\{ x_1,\ldots,x_n \} \subseteq Var(\Gamma)$. 
    $$
    \begin{array}{lll}
    & (\mathcal{M}^c, \nu^c), \Gamma \Vdash P(x_1, \dots, x_n)  &
    \\
    \Leftrightarrow & \langle \nu^c(x_1), \dots, \nu^c(x_n) \rangle \in I^c(P, \Gamma) & \text{(by definition of truth)}   \\  
    \Leftrightarrow & \langle x_1, \dots, x_n \rangle \in I^c(P, \Gamma) &    
    \text{ (by definition of } \nu^c) \\
    \Leftrightarrow & 
    P(x_1, \dots, x_n) \in Form(\Gamma) &  \text{ (by definition of }I^c)
    \\
 \end{array}
 $$
   
 \item $\Phi = \neg \Psi$.   
$$
\begin{array}{lll}
& (\mathcal{M}^c, \nu^c), \Gamma \Vdash \neg \Psi &  
\\ 
\Leftrightarrow & 
\nu^c(FV^h(\Psi))\subseteq Var(\Gamma) \mbox{ and }(\mathcal{M}^c, \nu^c), \Gamma \not \Vdash \Psi   &
\text{ (by definition of truth)}\\
\Leftrightarrow & \nu^c(FV^h(\Psi))\subseteq Var(\Gamma) \mbox{ and } \Psi \not \in Form(\Gamma)   &
\text{ (by induction hypothesis) }\\
\Leftrightarrow  &
\neg \Psi \in Form(\Gamma). & 
\text{ (by Lemma \ref{lm_max_sets})}\\
\end{array}
$$

 \item  $\Phi = \Psi \land \Theta$.  
 
$$
 \begin{array}{lll}
& (\mathcal{M}^c, \nu^c), \Gamma \Vdash \Psi \land \Theta  & 
\\ 
\Leftrightarrow & (\mathcal{M}^c, \nu^c), \Gamma \Vdash \Psi \text{ and } (\mathcal{M}^c, \nu^c), \Gamma \Vdash \Theta  &
\text{ (by definition of truth)}\\
\Leftrightarrow &  \Psi \in Form(\Gamma) \text{ and } \Theta \in Form(\Gamma) & 
\text{ (by induction hypothesis) }\\
\Leftrightarrow &\Psi \land \Theta \in Form(\Gamma).
 &  \text{ (by Lemma \ref{lm_max_sets})} \\
\end{array}
$$

 \item $\Phi = \forall x \Psi$. 
 
Suppose that $\forall x \Psi \in Form(\Gamma)$. By definition of the language $\cL^h(Var^+)$ we have $x\in Var$.  Then by Lemma \ref{lm_max_sets}
 $\Psi [x/v] \in Form(\Gamma)$ for all $v \in Var(\Gamma)$. Note that if $\nu^c$ is canonical then 
$(\nu^c)^x_v$ is canonical too, so by induction hypothesis for all $v \in Var(\Gamma) \ (\mathcal{M}^c, (\nu^c)^x_v), \Gamma \Vdash  \Psi$  whence $(\mathcal{M}^c, \nu^c), \Gamma \Vdash  \forall x \Psi$. 
 
 If $\forall x \Psi \not\in Form(\Gamma)$ then 
 $\neg\Psi [x/v] \in Form(\Gamma)$ for some $v \in Var(\Gamma)$ by $\exists$-completeness of $Form(\Gamma)$. By induction hypothesis $(\mathcal{M}^c, \nu^c), \Gamma \not\Vdash  \Psi[x/v]$  whence $(\mathcal{M}^c, \nu^c), \Gamma \not\Vdash  \forall x \Psi$. 

 \item $\Phi = \Box_X \Psi$.  
 
 $(\Leftarrow)$ 

Suppose that $\Box_X \Psi \in Form(\Gamma)$. 
By definition of $W^c$ formula
$\Box_X \Psi$ is $Var(\Gamma)$--closed, therefore 
$X\subseteq Var(\Gamma)=D^c_\Gamma$, so $\nu^c(X)=X\subseteq D^c_\Gamma$.
Assume that $\Delta \in W^c$ and $\Gamma R^c\Delta$. By definition of $R^c$ we have $\forall y_1 \ldots \forall y_n \Psi \in Form(\Delta)$ 
  where $\{y_1, \ldots, y_n \} = FV(\Psi)$. 
  Take arbitrary $v_1, \ldots, v_n \in Var(\Delta)$, then by first-order logic $\Psi[y_1/v_1, \ldots, y_n/v_n] \in Form(\Delta)$.  Then, by induction hypothesis, $(\mathcal{M}^c, \nu^c), \Delta \Vdash \Psi[y_1/v_1, \ldots, y_n/v_n]$. Therefore  $(\mathcal{M}^c, (\nu^c)^{y_1, \ldots, y_n}_{v_1, \ldots, v_n}), \Delta \Vdash \Psi$
  for all $\Delta \in R^c(\Gamma)$ and all $v_1, \ldots, v_n \in Var(\Delta)$, hence  $(\mathcal{M}^c, \nu^c), \Gamma \Vdash \Box_X \Psi$. 
 
 $(\Rightarrow)$ 
 
Assume that  $\Box_X \Psi$ is a $Var(\Gamma)$--closed formula and  $\Box_X \Psi\not \in Form(\Gamma)$. Then $X\subseteq Var(\Gamma)$ and $X$ does not contain variables from $Var$. Also 
	 $\neg\Box_X \Psi\in Form(\Gamma)$ since $Form(\Gamma)$ is $\cL^h(Var(\Gamma))$--maximal.
Take
 $$
 \Gamma^\#  := \{ \forall y_1 \ldots \forall y_n \Theta\ \mid\
\exists Z\subsetneq Var(\Gamma) (\Box_Z \Theta \in Form(\Gamma)) \text{ and } \{y_1, \ldots, y_n\} = FV(\Theta) \}
 $$
Take $S =  \Gamma^\# \cup \{\neg \Psi[y_1/v_1, \ldots, y_m/v_m] \},$ 
where $Y = \{y_1, \ldots, y_m \}=FV(\Psi)$ and $v_1, \ldots, v_m$ are distinct variables from $Var^+\setminus Var(\Gamma)$. Let us show that  $S$ is consistent. For contradiction assume that it is not, then there is a finite set of formulas $\{\Box_{Z_i}\Theta_i\ |\ i=1,\ldots,n\}$ from 
$Form(\Gamma)$ such that
\begin{equation}
\forall \overrightarrow{Y_1} \Theta_1, \dots ,\forall \overrightarrow{Y_n} \Theta_n, \neg\Psi [y_1/v_1, \ldots, y_m/v_m]\vdash\bot
\mbox{ where } 
\overrightarrow{Y_i}=FV(\Theta_i).
\end{equation}
 Using first-order logic one consequently derives in $\folpcs{0}(Var(\Gamma) \cup \{ v_1, \ldots, v_m  \})$
$$
\begin{array}{ll}
 (\forall \overrightarrow{Y_1} \Theta_1  \wedge \dots \wedge \forall \overrightarrow{Y_n} \Theta_n ) \to \Psi [y_1/v_1, \ldots, y_n/v_n] & 
\text{ by deduction theorem};\\ 
 (\forall \overrightarrow{Y_1} \Theta_1 \wedge \dots \wedge \forall \overrightarrow{Y_n} \Theta_n) \to \forall \overrightarrow{Y} \Psi  & \text{ by Bernays' rule}\\
\Box_\varnothing (\forall \overrightarrow{Y_1} \Theta_1 \wedge \dots \wedge \forall \overrightarrow{Y_n} \Theta_n \to \forall \overrightarrow{Y} \Psi ) & \text{ by necessitation rule}\\
\end{array}
$$
 Note that $\forall \overrightarrow{Y_i} \Theta_i$ for $i=1,\ldots,n$ are $Var(\Gamma)$--closed formulas and   
$FV^+(\forall \overrightarrow{Y_i} \Theta_i)=FV^+(\Theta_i)$. Put 
$X_i=FV^+(\Theta_i)$ and $\widetilde{X} = X\cup\bigcup\limits^n_{i=1} X_i$.
Then we continue derivation as follows
$$
\begin{array}{ll}
 \Box_{\widetilde{X}} (\forall \overrightarrow{Y_1} \Theta_1 \wedge \dots \wedge \forall \overrightarrow{Y_n} \Theta_n \to \forall \overrightarrow{Y} \Psi ) & \text{ by (A2$'$)}\\
(\Box_{\widetilde{X}} \forall \overrightarrow{Y_1} \Theta_1 \wedge \dots \wedge \Box_{\widetilde{X}} \forall \overrightarrow{Y_n} \Theta_n) \to \Box_{\widetilde{X}} \forall \overrightarrow{Y} \Psi&  \text{ by (A4$'$)};\\
 (\Box_{X_1} \forall \overrightarrow{Y_1} \Theta_1 \land \dots \land \Box_{X_n} \forall \overrightarrow{Y_n} \Theta_n) \to \Box_{\widetilde{X}} \forall \overrightarrow{Y} \Psi & \text{ by (A1$'$) }\Box_{X_i} \forall \overrightarrow{Y_i} \Theta_i \to \Box_{\widetilde{X}} \forall \overrightarrow{Y_i} \Theta_i;\\
 (\Box_{X_1}  \Theta_1 \land \dots \land \Box_{X_n}  \Theta_n) \to \Box_{\widetilde{X}} \forall \overrightarrow{Y} \Psi & \text{ since } \Box_{X_i}  \Theta_i\leftrightarrow\Box_{X_i} \forall \overrightarrow{Y_i} \Theta_i ;\\
\end{array}
$$
Taking into account that 
$X_i=FV^+(\Theta_i)$ and formulas $\Box_{Z_i}\Theta_i$ are $Var(\Gamma)$--close and thus do not contain bound occurrences of variables from  $Var(\Gamma)$, we conclude that $X_i\subseteq Z_i$. 
Then by axiom (A1$'$) we get
$\Box_{Z_i} \Theta_i\rightarrow \Box_{X_i} \Theta_i$.
We continue the derivation
$$
\begin{array}{ll}
 (\Box_{Z_1}  \Theta_1 \land \dots \land \Box_{Z_n}  \Theta_n) \to \Box_{\widetilde{X}} \forall \overrightarrow{Y} \Psi & \text{ since } \Box_{Z_i}  \Theta_i\to\Box_{X_i}  \Theta_i ;\\
 (\Box_{Z_1}  \Theta_1 \land \dots \land \Box_{Z_n}  \Theta_n) \to 
\Box_{X} \forall \overrightarrow{Y} \Psi & \text{ by } (A1')\text{ since }X\subseteq \widetilde{X}  ;\\
(\Box_{Z_1}  \Theta_1 \land \dots \land \Box_{Z_n}  \Theta_n) \to 
\Box_{X} \Psi & \text{ since } \Box_{X}  \Psi\leftrightarrow\Box_{X} \forall \overrightarrow{Y} \Psi.\\
\end{array}
$$
Therefore by 
Lemma \ref{lm_max_sets} $\Box_X \Psi \in Form(\Gamma)$,  contradiction. 

Thus, $S$ is $\cL(Var(\Gamma))$--consistent. We apply Lemma \ref{lm_extension} to  the consistent set of formulas $S$ and sets of variables $V_1= Var(\Gamma) \cup \{ v_1, \ldots, v_n \}$ and $V_2\supseteq V_1$ such that $V_2\setminus V_1$ is countable and the complement of $V_2$ to $Var^+$ is countable. It gives us the set of formulas $S^+\supseteq S$ which is  $\cL^h(V)$-maximal $\cL(V)$-consistent $\exists$--complete w.r.t. $V_2$. Put $\Delta = (S^+,V_2)$.  
By construction $\Delta\in R^c(\Gamma)$. 
Since $\neg\Psi[y_1/v_1,\ldots,y_n/v_n]\in S\subseteq S^+= Form(\Delta)$, by induction hypothesis we obtain $(\mathcal{M}^c, \nu^c ), \Delta \not \Vdash \Psi[y_1/v_1,\ldots,y_n/v_n]$. Therefore, $(\mathcal{M}^c, \nu^c), \Gamma \not \Vdash \Box_X \Psi$. 

\item $\Phi := \term{t}{X} \Psi$.

$(\Leftarrow)$ 

Suppose that $\term{t}{X} \Psi \in Form(\Gamma)$. Firstly, by definition of $\nu^c$ and $W^c$ we have $\nu^c (X)=X\subseteq Var(\Gamma)=D^c_\Gamma$. 
Secondly, $\Box_X \Psi \in Form(\Gamma)$
and we repeat the proof, given above, to show that $(\mathcal{M}^c, (\nu^c)^{y_1, \ldots, y_n}_{v_1, \ldots, v_n}), \Delta \Vdash \Psi$
  for all $\Delta \in R^c(\Gamma)$ and all $v_1, \ldots, v_n \in Var(\Delta)$.
It remains to show that $\Gamma\in\cE^c(t,\Psi,\nu^c\upharpoonright(FV^h(\Psi)\cap X))$, that is, 
$$
\begin{array}{c}
Im(\nu^c\upharpoonright(FV^h(\Psi)\cap X))\subseteq Var(\Gamma)\ \mbox{ and }\\
(\term{t}{Y}\Psi) (\nu^c\upharpoonright(FV^h(\Psi)\cap X))\in Form(\Gamma)\ \mbox{ for }
Y=Dom(\nu^c\upharpoonright(FV^h(\Psi)\cap X))\cap FV^h(\Psi) \\
\end{array}
$$ 
Note that $Dom(\nu^c\upharpoonright(FV^h(\Psi)\cap X)) = FV^h(\Psi)\cap X = FV^+(\Psi)\cap X \subseteq Var^+$.
Since $\nu^c$ is canonical, $\nu^c(x)=x$ for each $x\in X$, whence the first condition is true. Then, 
$Y=Dom(\nu^c\upharpoonright(FV^h(\Psi)\cap X))\cap FV^h(\Psi) = FV^+(\Psi)\cap X$ and $\nu^c$ is identical on $Var^+$, so 
$(\term{t}{Y}\Psi) (\nu^c\upharpoonright(FV^h(\Psi)\cap X))$ coincides with 
$(\term{t}{FV^+(\Psi)\cap X}\Psi)$. So, by Axiom (A1) and Lemma \ref{lm_max_sets}
from $(\term{t}{X}\Psi)\in Form(\Gamma)$ we get $(\term{t}{FV^+(\Psi)\cap X}\Psi)\in Form(\Gamma)$ and the second condition is true too.

$(\Rightarrow)$ 

Assume that $\term{t}{X} \Psi \not \in Form(\Gamma)$. 
Let us show that 
$\Gamma\not\in\cE^c(t,\Psi,\nu^c\upharpoonright FV^h(\Psi)\cap X)$, namely, that 
\begin{equation}
\label{eq100}
(\term{t}{Y} \Psi)(\nu^c\upharpoonright FV^h(\Psi)\cap X)\not\in Form(\Gamma)
\mbox{ for } Y=FV^h(\Psi)\cap Dom(\nu^c\upharpoonright FV^h(\Psi)\cap X).
\end{equation}
Since $\term{t}{X} \Psi \not \in Form(\Gamma)$ is a $V$--closed formula, 
$X\subseteq Dom(\Gamma)\subseteq V$. 
Therefore $\nu^c(x)=x$ for all $x\in X$ and
$$
\begin{array}{c}
Y=FV^h(\Psi)\cap Dom(\nu^c\upharpoonright FV^h(\Psi)\cap X)=
FV^h(\Psi)\cap  X \subseteq X\\
\nu^c\upharpoonright FV^h(\Psi)\cap X \mbox{ is identical }\\
\end{array}
$$
Therefore formulas  
$(\term{t}{Y} \Psi)(\nu^c\upharpoonright FV^h(\Psi)\cap X)$ and 
$\term{t}{Y} \Psi$ coincide. 
Since $X\subseteq Y$, if $\term{t}{Y} \Psi\in Form(\Gamma)$ then by Axiom $(A2)$
$\term{t}{X} \Psi\in Form(\Gamma)$, contradiction.
\end{itemize}
\end{proof}

\subsection{Proof of the Strong Completeness}
Now we are ready to prove Theorem \ref{st_compl}. 

\begin{proof}
    Suppose that $CS$ is a constant specification, $\Gamma$ is a set of closed formulas and $\Phi$ is a closed formula in language $\cL$. Suppose that $\Gamma \not\vdash_{CS} \Phi$. Then $\Gamma \cup \{ \neg \Phi\}$ is a $\cL$--consistent using $CS$ set of closed formulas. 
    
    Consider the canonical model $\cM^c$ of Definition \ref{can_mod}.
    By Lemma \ref{lm_canonical_model}, it is  a Fitting model meeting $CS$. 
    
    Let $V$ be a countable subset of variables from $Var^+$ having countable complement to $Var^+$. Then by Lemma \ref{lm_extension}, there exists a set of $V$-closed formulas $S^+$ such that 
    $S^+$ is $\cL^h(V)$--maximal, $\cL(V)$--consistent using $CS$, $\exists$--complete w.r.t. $V$ and formulas from
    $\Gamma \cup \{ \neg \Phi\}$ belong to $S^+$. Take $\Delta = (S^+,V)$, then  $\Delta \in W^c$. Thus, $\Gamma \subseteq Form(\Delta)$ and $\Phi \not \in Form(\Delta)$. Since formulas $\Gamma$ and $\Phi$ are closed formulas of the original language $\cL$, they all are $V$-closed, therefore  $(\mathcal{M}^c, \nu^c ), \Delta \Vdash \Gamma$ and $(\mathcal{M}^c, \nu^c), \Delta\Vdash\neg \Phi$ due to the Lemma \ref{lm_truth}. 
\end{proof}

\subsubsection*{Acknowledgement}
We are grateful to S.N. Artemov for encouragement and discussions during our work. We also like to thank V.B. Shehtman and D.Shamkanov for the possibility to present our work at their research seminars and many valuable comments and to V.N. Krupski who was the first reader of our manuscript and found many confusions and mistakes. 

This work was supported by the Russian Science Foundation under grant no. 23-11-00104, \url{https://rscf.ru/en/project/23-11-00104/ }.

\bibliographystyle{apacite}
\bibliography{all_references}

\end{document}